\documentclass[a4paper,12pt]{article}

\usepackage[a4paper,
textwidth = 16cm,
textheight = 22.5cm,
top = 3cm]{geometry}

\usepackage{amsmath,amsthm,amssymb,bbm,graphicx,xcolor,natbib,url,enumitem}
\usepackage{float}
\usepackage[ruled, vlined, linesnumbered]{algorithm2e}
\usepackage[colorlinks,citecolor=blue,urlcolor=blue,breaklinks]{hyperref}

\DeclareMathOperator*{\argmax}{arg\,max}
\DeclareMathOperator*{\argmin}{arg\,min}
\let\leq=\leqslant   
\let\geq=\geqslant
\newtheorem{theorem}{Theorem}
\newtheorem{prop}[theorem]{Proposition}

\newtheorem{corollary}[theorem]{Corollary}
\title{SpreadDetect: Detection of spreading change in a network over time}
\author{Hanqing Cai and Tengyao Wang\footnote{Research supported by EPSRC grant EP/T02772X/1.}\\
University College London and London School of Economics}
\date{\today}

\newcommand{\RN}[1]{%
  \textup{\uppercase\expandafter{\romannumeral#1}}%
}
\begin{document}
\maketitle
\begin{abstract}
Change-point analysis has been successfully applied to the detect changes in multivariate data streams over time. In many applications, when data are observed over a graph/network, change does not occur simultaneously but instead spread from an initial source coordinate to the neighbouring coordinates over time. We propose a new method, SpreadDetect, that estimates both the source coordinate and the initial timepoint of change in such a setting. We prove that under approriate conditions, the SpreadDetect algorithm consistently estimates both the source coordinate and the timepoint of change and that the minimal signal size detectable by the algorithm is minimax optimal. The practical utility of the algorithm is demonstrated through numerical experiments and a COVID-19 real dataset. 
\end{abstract}

\section{Introduction}
The advance of technology has allowed us to collect vast amount of time-ordered data. A common phenomenon in such datasets is that the data generating mechanism may change over time.  Examples include climate data that tracks the amount of greenhouse gases in the atmosphere \citep{Reeves2007ARA}; \citep{IJ2010} functional Magnetic Resonance Imaging (fMRI) scans that record the time evolution of blood oxygen level dependent (BOLD) chemical contrast in different areas of the brain \citep{aston2012};\citep{Bosc2003} and virtually  simultaneous market shocks in financial  data  stream \citep{ChenGupta1997}. The presence of changes renders traditional statistical techniques that rely on the homogeneity of the dataset to be inapplicable. A popular way of handling the inhomogeneity caused by changes in such datasets is through the technique of change-point analysis, where we detect and localise timepoints of change so as to segment the original data series into shorter segments that are more or less stationary. 

Classically, change-point analysis focused primarily on the setting of univariate time series, with some state-of-art methods proposed in \citet{KFE2012, FMS2014, Fryzlewicz2014}. However, these methods are often inadequate for modern, high-dimensional data-sets, where signals may be spread across multiple coordinates. Recently,   
several methodologies have been proposed to test and estimate change-points in the high-dimensional settings by borrowing strength across multiple coordinates to detect and localise change-points at a higher accuracy than would otherwise be possible using univariate change-point algorithms alone. These methods include $\ell_2$ or $\ell_\infty$ aggregation of the cumulative sums (CUSUMs) test statistics across different components proposed by  \citet{HorvathHuskova2012, Jirak2015}, the Sparsified Binary Segmentation algorithm by \citet{ChoFryzlewicz2015}, the double CUSUM algorithm of \citet{Cho2016} and a projection-based approach by \citet{wang2016highdimensional}. 
 
However, in order to handle the high-dimensional nature of the problem, the multivariate or high-dimensional methods mentioned in the previous paragraph often make simplifying assumptions such as all coordinates are exchangeable or that changes are located in a sparse subset of coordinates. In reality, in many applications, there are additional structures in the change-points that one can exploit to improve the estimation accuracy. Examples include group structures where coordinates form natural groups and changes tend to occur within the same group \citep{wang2021statistically,CW2023}, and community structures where nodes belong to different (unknown) communities and may switch community at the change-point \citep{wang2021optimal}. In the present work, we consider the structure where the coordinates represents nodes of a graph/network and the change, instead of occurring simultaneously in all coordinates of interest, may initially appear in one coordinate (the \emph{source coordinate} of change), and then spread across the network gradually over time. Such a statistical model is useful to represent, for instance, the spread of infectious disease between individuals over time. We are interested to estimate both the source coordinate and the timepoint where the change occurs at the source coordinate. Note that different coordinate will have a change occurring at a different timepoint. To avoid ambiguity, we refer to the time of the change in the source coordinate as the \emph{initial change-point}, or simply the \emph{change-point} of the model, and the timepoint of change in any given coordinate as the \emph{time of spread} to that coordinate, which is typically later than the change-point. In such as setting, the change signal may be very small and sparse when first appears, and increases as the change is spread across the network. Thus, a naive application of a multivariate change-point procedure may miss the initial part of the change and likely estimate a change-point with a positive bias. Moreover, in many applications, the coordinate(s) where the change first appears may be of separate interest. This calls for a new methodology that can handle the spreading nature of the change.

In this paper, we proposed a method, called SpreadDetect, that deals with the task of estimating both the source coordinate and the initial change-point time in a statistical model where the change is spread across the network via adjacent nodes. The key idea here, is to aggregate evidence of change, measured in terms of coordinatewise CUSUM statistics, across multiple coordinates with suitable time lags. We then centre these aggregated CUSUM statistics so that under the null distribution, candidate change-points near and far away from the boundary of the time window considered are treated on equal footings. The method is explained in detail in Section~\ref{sec:method}. Depending on whether the signs of the change in different coordinates are equal, we propose quadratic and linear test statistics respectively, indexed both in time and over the coordinates. The final estimator for the time and coordinate of initial change is obtained by maximising these aggregated statistics. 

In Section~\ref{Sec:Theory}, we derive theoretical guarantees of our proposed SpreadDetect method. For simplicity, we focus on the case where the change is spreading across the network at a deterministic rate. We assume that if the change-point and source coordinate pair varies from $(z^*, j^*)$ to $(t^*, k^*)$, at least $m$ nodes in the network will witness a difference in their time of spread at least proportional to the sum of the time difference between $z^*$ and $t^*$ and the graph distance between $j^*$ and $k^*$. 
 We first derive a key result in Theorem~\ref{Thm:GeneralConsistency}, saying that assuming that the change is bounded away from the endpoint, and provided the magnitude of change is up to logarithmic factors above $\sqrt{p}/(n  m) + p/(n  m^2)$,  then both the source coordinate and initial change-point time can be accurately estimated. Theorem~\ref{thm:testing} then shows that our estimation procedure can be turned into a test with good size and power controls for testing the existence of a change-point of the above signal size. 
 Theorem~\ref{Thm:LowerBound} shows that when $m \asymp p$ (a condition that can be verified in many common graphs),  the signal size required in Theorem~\ref{Thm:GeneralConsistency} is in fact minimax optimal. In addition, we derive in Theorem~\ref{Thm:LinearStatConsistency} the result for the special case when we know the sign of the signal so that the linear statistics in Algorithm~\ref{alg:split} is used.  In this case, the estimation accuracy is guaranteed if the magnitude of change is above $1/\sqrt{mn\tau^2}$ up to logarithmic factor. 

In Section~\ref{Sec:Simulations}, we evaluate the empirical performance of method through simulated data and a COVID-19 real data example. We evaluate our method under two settings when the signal spread to the nearby coordinates in a fixed or random way using the simulated data. Proofs of all theoretical results are deferred to Section~\ref{Sec:Proofs}, and ancillary results and their proofs are given in Appendix~\ref{Sec:Ancillary}.
We conclude our introduction with some notation used throughout the paper.

\subsection{Notation}
For $n\in\mathbb{N}$, we write $[n] = \{1, \dots, n\}$. For a vector  $\|v\|_q=\bigl\{\sum_{i=1}^{n} (v_i)^q\bigr\}^{1/q}$ for any positive integer $q$. We denote $j = \lceil p \rceil$ if $j$ is the smallest integer such that $j\geq p$ and denote $j = \lfloor p \rfloor$ if $j$ is the largest integer such that $j\leq p$.

Given two sequences $(a_n)_{n\in\mathbb{N}}$ and $(b_n)_{n\in\mathbb{N}}$ such that $a_n, b_n > 0$ for all $n$, we write $a_n \lesssim b_n$ (or equivalently $b_n \gtrsim a_n$) if $a_n \leq C b_n$ for some universal constant~$C$. We write $a_n \asymp b_n$ if $0<\liminf_{n\rightarrow \infty} |a_n/b_n|\leq \limsup_{n\rightarrow \infty} |a_n/b_n|<\infty$.


\section{Problem setup and methodology}
\label{sec:method}
Given a network represented by a connected graph $G$, with vertices $V(G) := [p]$ and edges $E(G) \subseteq [p]\times [p]$,  let $j^*\in V(G)$ be the source coordinate and $z^*\in[n]$ the change-point and write $S_t\subseteq [p]$ for the set of ``infected nodes'', i.e., coordinates that have undergone a change at or before time $t$. We have  $S_t = \emptyset$ for $t < z^*$, $S_{z^*} = \{j^*\}$ and we assume that the change spreads from infected nodes to their neighbours at a constant rate in the sense that at any time $t > z^*$, $S_t := \{j: \text{$(j,k) \in E(G)$ for some $k\in S_{t-1}$}\}$. Suppose the data $X_1,\ldots,X_n \in \mathbb{R}^{V(G)} \cong \mathbb{R}^p$ follow multivariate normal distribution with an identity covariance such that  
\[
\mathbb{E}(X_t) = \mu^0 \circ \mathbf{1}_{S_t^{\mathrm{c}}} + \mu^1 \circ \mathbf{1}_{S_t}, \quad \text{for $t\in [n]$},
\]
where $\circ$ denotes the Hadamard product, $\mu^0$ and $\mu^1$ are respectively vectors of means pre- and post-change, and $\mathbf{1}_{A} := (\mathbbm{1}_{j\in A})_{j\in[p]}$ for any $A\subseteq [p]$.

Let $d_G(j,k)$ be the graph distance between nodes $j$ and $k$, i.e., the length of the shortest path from $j$ to $k$ on graph $G$. Then, the data consist of independent random variables
\[
X_{j,t}\sim
\begin{cases}
  N(\mu_j^0,1)& \text{ if $t\leq z^*+d_G(j,j^*)$} \\ N(\mu_j^1,1)& \text{ if $t>z^*+d_G(j,j^*)$,}
  \end{cases} \quad \text{for $j\in[p]$ and $t\in[n]$}.
\]
We define $P_{j^*,z^*,\mu^0,\mu^1}$ to be the distribution of the data matrix $X = (X_1,\ldots,X_n)\in\mathbb{R}^{p\times n}$ given parameters $(j^*, z^*, \mu^0, \mu^1)\in\Theta := [p]\times [n]\times \mathbb{R}^p\times \mathbb{R}^p$. Our task is to estimate $j^*$ and $z^*$ given data $X \sim P_{j^*,z^*,\mu^0,\mu^1}$. Define $\theta = (\theta_1,\ldots,\theta_p)^\top := \mu^1 - \mu^0$. We assume that $|\theta_j|\geq a$ for some $a>0$ for all $j\in[p]$. 

Writing $\mu = \mathbb{E}X\in\mathbb{R}^{p\times n}$, we have the decomposition $X=\mu +W$, where $W$ is a $p\times n$ random matrix with independent $N(0,1)$ entries. Define $\mathcal{T}: \mathbb{R}^{p\times n}\to \mathbb{R}^{p\times (n-1)}$ to be the matrix CUSUM transformation:
\begin{equation}
\label{Eq:CUSUM}
  [\mathcal{T}(M)]_{j,t}=\sqrt{\frac{t(n-t)}{n}}\bigg(\frac{1}{n-t}\sum_{r=t+1}^nM_{j,r}-\sum_{r=1}^t\frac{1}{t}M_{j,r}\bigg), 
\end{equation} 
we form $T=\mathcal{T}(\mathbf{X})$. 

The CUSUM transformation is the normalised difference before and after the change for a single entry in the data matrix. Motivated by this meaning, and that in each coordinate, the CUSUM statistics is maximised at the time of spread, we propose to aggregate these CUSUM statistics in different coordinates at appropriate lags. Specifically, given any candidate source coordinate and change-point pair $(j,t)$, we compute the time of spread to each coordinate $k$ as $t+d_G(j,k)$ and aggregate $T_{k,t+d_G(j,k)}$ over $k\in[p]$ provided that $t+d_G(j,k) \leq n$.  For each $t\in [n-1]$ and $k\in[p]$, we define $\mathcal{J}_{j,t}:=\{k\in[p]: t+d_G(j,k) < n\}$.  If we do not know the sign of the signal, We form the following quadratic statistics
\begin{equation}
\label{Eq:quadratic}
Q_{j,t} := \sum_{k\in \mathcal{J}_{j,t}} (T_{k, t+d_G(j,k)}^2-1).    
\end{equation}
Here, we subtract $1$ from the summands to make them mean-centred, so that candidate change-points near the right boundary will not be disfavoured due to the set $\mathcal{J}_{j,t}$ being smaller.  

We then estimate the location of the change-point $z^*$ and the source coordinate of the spread via
\begin{equation}
\label{Eq:QuadraticEstimator}
 (\hat j, \hat z) = \argmax_{j,t} Q_{j,t}.
\end{equation}

Practically, the $p\times p$ distance matrix $(d_G(j,k):j,k\in[p])$ between every pair of vertices can be pre-calculated from the adjacency matrix in $O(p^3)$ time using the Floyd--Warshall algorithm \citep{floyd1962algorithm}. The entire estimation procedure is summarised in Algorithm~\ref{alg:split}.

\begin{algorithm}[htbp]
 \caption{Spreading change estimation procedure }
 \label{alg:split}
    \KwIn{$X\in \mathbb{R}^{p\times n}$, graph $G$}
    Compute $T \leftarrow \mathcal{T}(X)$ as in~\eqref{Eq:CUSUM}.\\
    Compute $Q_{j,t}$ for $j\in[p]$ and $t\in[n-1]$ via Equation~\eqref{Eq:quadratic}. \\
    Estimate $(\hat j,
    \hat z) = \argmax_{j,t} Q_{j,t}$.\\
    \KwOut{$(\hat j, \hat z)$}
\end{algorithm}

Figure~\ref{fig:CUSUM} illustrates the working of Algorithm~\ref{alg:split} in action.
Here, we have a data matrix $X\in\mathbb{R}^{200\times200}$, which contains a change spreading from the source coordinate $j^*=50$ from the changepoint time $z^*=50$. The right panel displays the matrix $(Q_{j,t})_{j\in[p], t\in[n-1]}$ of aggregated squared CUSUM statistics from equation~\ref{Eq:quadratic}. The darker colour indicates larger values of the $Q_{j,t}$ statistics. We can see that the aggregation proposed by~\ref{Eq:quadratic} indeed helps us locate both the source coordinate and the true time of change-point.

\begin{figure}[htbp]
\centering
\begin{tabular}{cc}
\includegraphics[width=8cm,height=8cm]{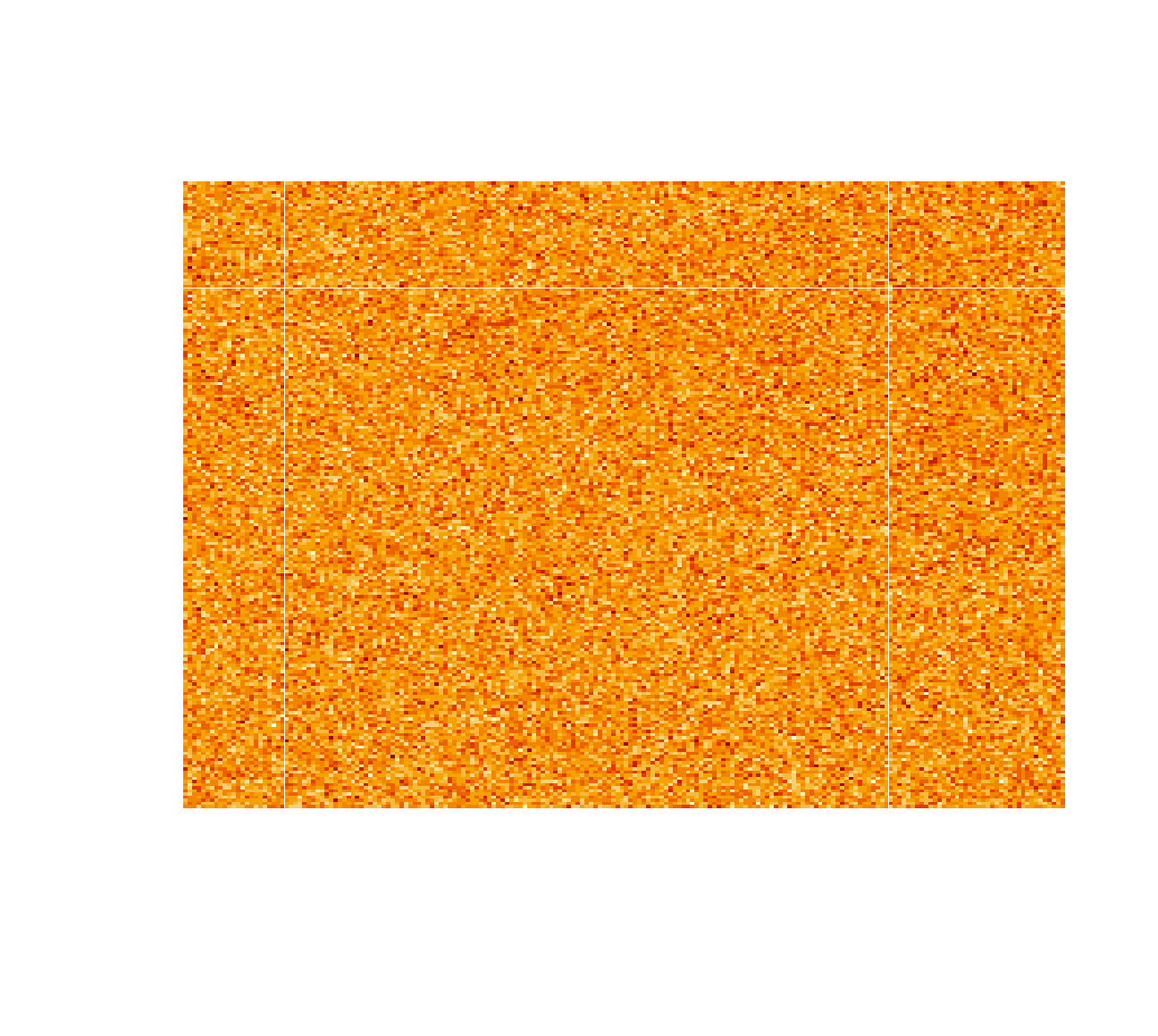} &
\includegraphics[width=8cm,height=8cm]{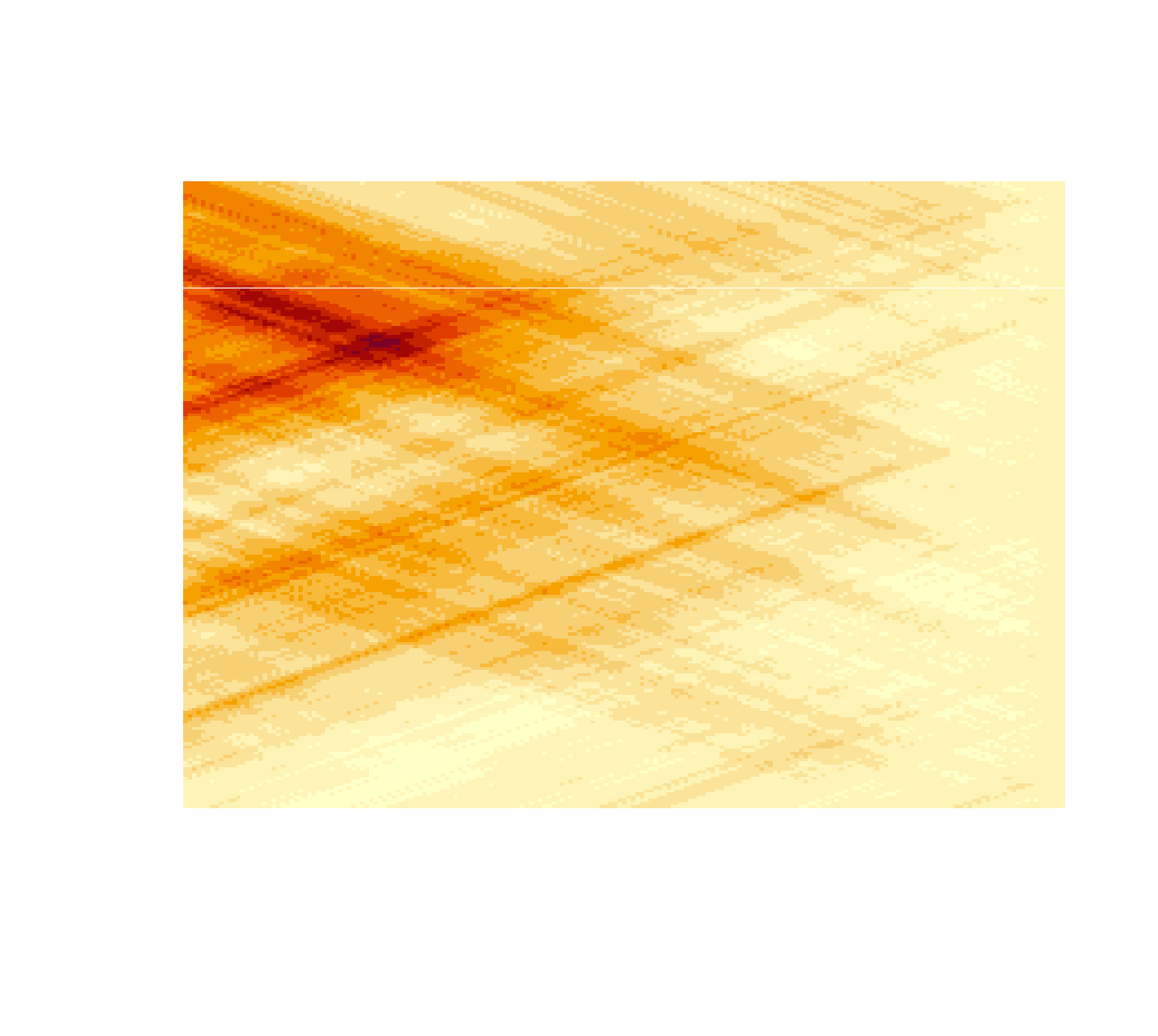}
\end{tabular}
\caption{\label{fig:CUSUM}Illustration of the SpreadDetect algorithm. The heatmap of the original data matrix $X$ is shown on the left panel, where data consist of $p=200$ nodes in a cycle graph measured over a period of $n=200$ time points. A true change occurs at $z = 50$ from coordinate $50$ and spread across the graph following the model described in Section~\ref{Sec:Simulations}. The right panel depicts the heatmap of the aggregated CUSUM statistics generated in the SpreadDetect algorithm. The estimated time of change $\hat z = 52$ and the estimated origin of change $\hat j = 45$ is where the matrix of the aggregated CUSUM statistics achieves its maximum value. }
\end{figure}

In some practical applications, it is reasonable to assume additionally that sign of the changes are the same across all coordinates. In such settings, we can modify the quadratic aggregation proposed in~\eqref{Eq:quadratic} by using the following linear statistics instead:

\begin{equation}
\label{Eq:linear}
L_{j,t} := \biggl|\sum_{k\in \mathcal{J}_{j,t}} T_{k, t+d_G(j,k)}\biggr|.
\end{equation}
The source coordinate and the change-point are then correspondingly estimated via
\begin{equation}
\label{Eq:LinearOptimiser}
    (\hat j, \hat z) = \argmax_{j,t} L_{j,t}
\end{equation}

\section{Theoretical results}
\label{Sec:Theory}
In this section, we derive theoretical guarantees of the change-point estimation procedure proposed in Algorithm~\ref{alg:split}. 

For any fixed $t^*$, $k^*$, we define the following set:
\begin{equation}
\label{Eq:mathcalJ}
\mathcal{J}_{t^*,k^*}(C_1)=\bigg\{j\in V(G): |z^*+d_G(j,j^*)-(t^*+d_G(j,k^*))|\geq C_1 (|z^*-t^*|+d_G(j^*,k^*))\bigg\}.
\end{equation}
We remark that $\mathcal{J}_{t^*,k^*}(C_1)$ also depends on $z^*$ and $j^*$, though we will suppress this dependence in the notation since in what follows, we will mostly treat $z^*$ and $j^*$ as fixed or can be inferred from the context. 
\begin{theorem}
\label{Thm:GeneralConsistency}
Suppose $n\tau\geq 2p$ and $X\sim P_{j^*,z^*,\mu_0,\mu_1}$ with $\mu_0-\mu_1 \in\{-a,a\}^p$. Define $m = m_G(C_1) := \min_{t^*,k^*} |\mathcal{J}_{t^*,k^*}(C_1)|$. There exists a universal constant $c>0$ such that if
\begin{equation}
\label{Eq:CondSignalGeneral}
a^2 \geq c\biggl\{\frac{\sqrt{p}+\log(2pn)}{n\tau m} + \frac{p\log(2pn)}{n\tau^2 m^2}\biggr\}.
\end{equation}
then, the estimator $(\hat j, \hat z)$ from~\eqref{Eq:QuadraticEstimator} satisfies with probability at least $1-1/(2pn)$ that 
\[
|\hat z - z^*| + d_G(\hat j , j^*) \leq \frac{12\sqrt{6}}{C_1 m}   \biggl\{ \frac{\sqrt{p}+\log(2pn)}{a^2} + \frac{\sqrt{pn \log(2pn)}}{a}\biggr\}.
\]
\end{theorem}

The $n\tau\geq 2p$ condition is placed to ensure that the change happens early in the time series to allow sufficient time to spread to all nodes in the network. This helps simplify our analysis and presentation. However, we note that a similar result can be derived without this assumption; see Theorem~\ref{Thm:Consistency} in Appendix~\ref{Sec:Ancillary}.
We remark also that Condition~\eqref{Eq:CondSignalGeneral} is mild in view of the conclusion of Theorem~\ref{Thm:GeneralConsistency}. Indeed, for the right-hand side of the loss bound to be nontrivial (i.e.\ less than $n+p$), we would at least need $a^2 \gtrsim \{\sqrt{p} + \log(2pn)\}/(nm) + p\log(2pn) / (nm^2)$. Thus,~\eqref{Eq:CondSignalGeneral} only requires $a^2$ to be larger than a factor of at most $\tau^{-2}$ than minimally what is required in Theorem~\ref{Thm:GeneralConsistency}. The final loss bound is inversely proportional to $C_1 m_G(C_1)$. In general, $m_G(C_1)$ is a decreasing function of $C_1$ and by the triangle inequality, $m_G(C_1) = 0$ for all $C_1\geq 1$. Hence, the optimal loss bound we can obtain involves a carefully chosen trade-off between $C_1$ and $m_G(C_1)$ in the denominator of the final bound. In practice, in many applications, we have $m_G(C_1) \asymp p$ for some $C_1\asymp 1$. Under such assumptions,  and if in addition $\log(n) = O(\sqrt{p})$, the conclusion of Theorem~\ref{Thm:GeneralConsistency} simplifies to
\[
\frac{|\hat z - z^*|}{n} + \frac{d_G(\hat j , j^*)}{p} = O\biggl(\frac{1}{p^{1/2}\|\theta\|_2^2 } + \frac{\sqrt{n\log(2pn)}}{p\|\theta\|_2 }\biggr),
\]
showing that both the location of change and the origin of change estimators are consistent when $\|\theta\|_2 \gg \max\{p^{-1/4}, p^{-1}\sqrt{n\log(2pn)}\}$. 

As mentioned above, the quantity $m_G(C_1)$ plays an important role in our theoretical control of the loss of change-point location and origin estimation. To get a sense of the magnitude of this quantity, we compute $m_G(1/4)$ for grid graphs, binary trees and random Erd\H{o}s--R\'enyi graphs. Figure~\ref{fig:proportion} shows that we have $m_G(1/4) \geq cp$ for some constant $c>0$ in all these simulation settings. Moreover, for each specific type of graph, $m_G(1/4)/p$ tends to be relatively stable when $p$ is large. Theoretically, $m(C_1)$ needs to be controlled in a case-specific manner. Below, we illustrate how this can be done in the setting of a $d$-dimensional grid graph. For simplicity of exposition, we introduce additional symmetry to require that the grid is `wrapped around the edges', in the sense that $G = \prod_{r=1}^d G_r$, where each $G_r$ is a $p_1$-cycle $C_{p_1}$ with $p_1^d = p$. Working with product of cycles instead of paths makes all vertices of $G$ equivalent. The following proposition controls $m_G(1/(4d))$ of such a graph $G$. 

\begin{figure}[tbp]
\centering
\includegraphics[width=11cm]{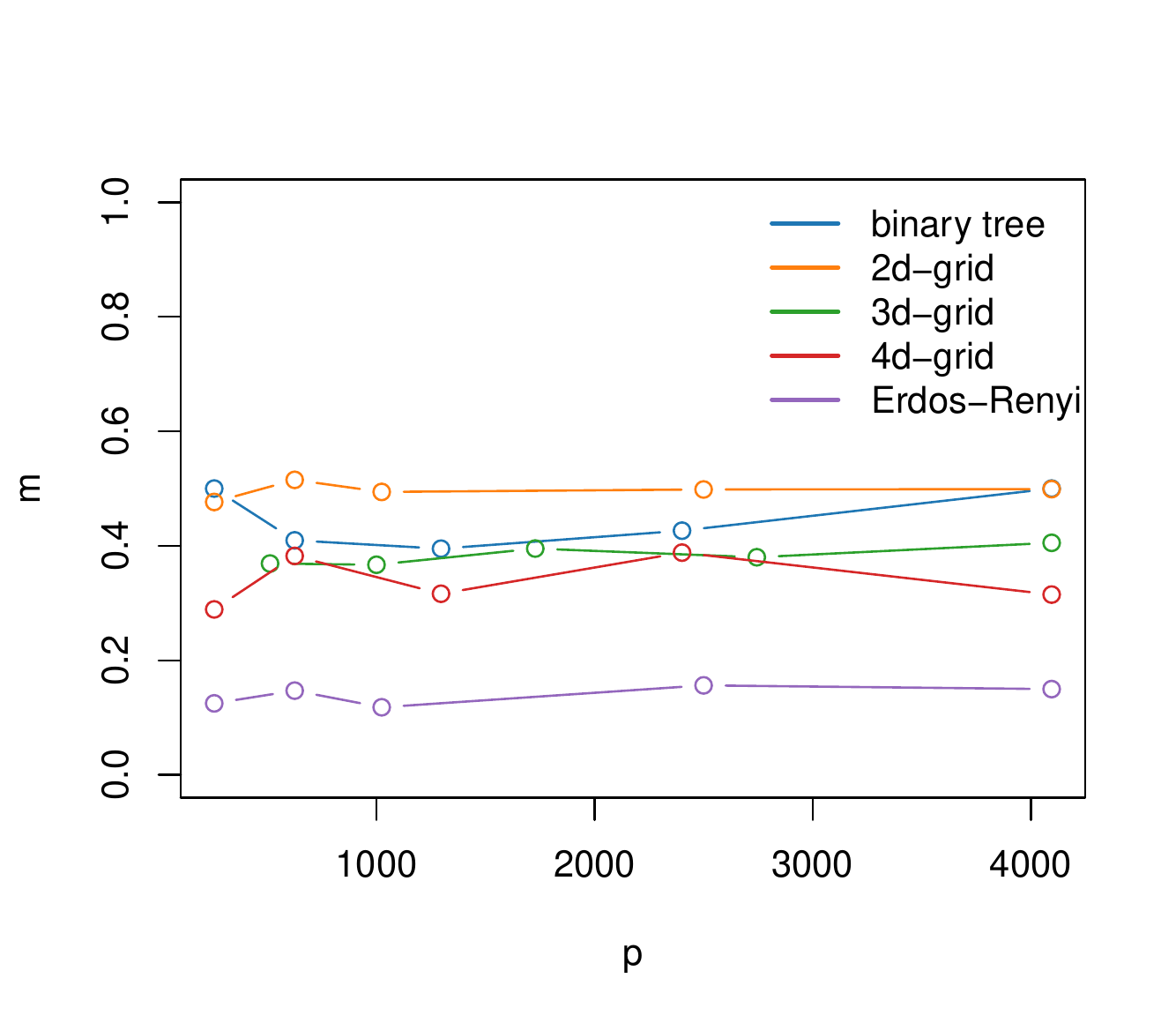} 
\caption{\label{fig:proportion}$m_G(1/4)/p$ for different graphs.}
\end{figure}

\begin{prop}
\label{Prop:gridcardinality}
Suppose $G = \prod_{r=1}^d G_r$ with $G_r \cong C_{p_1}$ for all $r\in[d]$ and $p=p_1^d$. Assume further that $n\tau\geq 2p_1$. Then we have $m_G(1/(4d)) \geq p/8^d$.
\end{prop}

Treating the dimension of the grid as fixed, we have the desired bound that $m_G(C_1)\asymp p$ for some $C_1\asymp 1$. The following result is an immediate consequence of Theorem~\ref{Thm:GeneralConsistency} together with Proposition~\ref{Prop:gridcardinality}.
\begin{corollary}
\label{Cor:Grid}
Under the same asumption as in Proposition~\ref{Prop:gridcardinality}. Suppose $X\sim P_{j^*,z^*, \mu_0,\mu_1}$ with $\mu_0-\mu_1\in\{-a,a\}^p$. There exist $c, C>0$, depending only on $d$, such that if 
\begin{equation}
\label{Eq:CondSignal}
a^2 \geq c\biggl\{\frac{\sqrt{p}+\log(2pn)}{n\tau p} + \frac{pn\log(2pn)}{n^2\tau^2 p^2}\biggr\},
\end{equation}
then with probability at least $1-1/(2pn)$, the estimator $(\hat j, \hat z)$  defined in~\eqref{Eq:QuadraticEstimator} satisfies that 
\[
|\hat z - z| + d_G(\hat j , j) \leq C \biggl\{ \frac{\sqrt{p}+\log(2pn)}{a^2p} + \frac{\sqrt{pn \log(2pn)}}{ap}\biggr\}.
\]
\end{corollary}

While the focus of our discussion so far has been the estimation of a changepoint (both in terms of the time of change and location of the source of change), our method can be easily modified for the related testing problem.  More preicsely, given the data $X$ described in Section~\ref{sec:method}, we are interest in testing $H_0: \theta=0$ against the alternative: $H_1: \theta\neq 0$. We can construct a test based on the quadratic statistics computed according to Algorithm~\ref{alg:split} as follows:
\begin{equation}
\label{Eq:Test}
\psi_\lambda(X)=\mathbbm{1}_{\{\max_{j\in[p],t\in[n-1]} Q_{j,t} \geq \lambda\}}.
\end{equation}
The following theorem shows that for an appropriate choice of $\lambda$, the test $\psi_{\lambda}$ defined above has small Type I and Type II errors.
\begin{theorem}
\label{thm:testing}
Given $X\sim P = P_{j^*,z^*,\mu_0,\mu_1}$. For any $\delta\in(0,1)$ and $\lambda\geq 2\sqrt{p\log(pn/\delta)}+2\log(pn/\delta)$, the test $\psi_{\lambda}$ defined in~\eqref{Eq:Test} has that following properties.
\begin{enumerate}[label=(\alph*)]
    \item If $\theta=0$, then 
    \[
    \mathbb{P}_{P}(\psi_\lambda(X)=1)\leq \delta. 
    \]
    \item There exists a universal constant $C>0$ such that if $a^2 \geq \frac{C\lambda}{n\tau\min\{2p,n\tau\}}$, then 
    \[
    \mathbb{P}_{P}(\psi_\lambda(X)=1)\geq 1-\delta.
    \]
\end{enumerate}
\end{theorem}

From Theorem~\ref{thm:testing} above, if $p=O(n\tau)$ and taking $\delta=1/(pn)$, then the test $\psi_{\lambda}$ defined in~\eqref{Eq:Test} is able to detect a change when $a^2 \geq \frac{C\sqrt{\log(pn)}}{\sqrt{p}n\tau}$. Note that when $m_G(C_1)$ is of order $p$ and $\sqrt{p}\tau \gtrsim \log(2pn)$, then the signal-size condition in~\eqref{Eq:CondSignalGeneral} is equivalent to 
\[
a^2 \gtrsim \frac{1}{n\tau\sqrt{p}} + \frac{\log(2pn)}{n\tau^2 p} \gtrsim \frac{1}{n\tau\sqrt{p}}.
\]
Hence, the signal strength needed here for testing is consistent, up to logarthmic factors, with ~\eqref{Eq:CondSignalGeneral} in Theorem~\ref{Thm:GeneralConsistency} in such a setting. However, the estimation problem is harder comparing to the testing problem, when $C_1 m_G(C_1)$ is much smaller than $p$ for all choices of $C_1 \in(0,1)$. This can happen, for instance, in the case when there exists a $t^*$ close to $z^*$ such that the signal from $z^*$ needs to pass from $t^*$ to spread over the rest of the coordinate,  it is hard to tell which time point does the signal starts. However, for the testing problem, we only need to know whether there is a change regardless of the location.

To understand the optimality of the signal-size condition in~\eqref{Eq:CondSignalGeneral}, we derive a minimax lower bound for testing the existence of a change-point. Let 
\begin{align*}
 \Theta_0 &:= \{(j^*,z^*,\mu_0,\mu_1)\in\Theta: \mu_0=\mu_1, \min(z^*,n-z^*)\geq n\tau\}\\
 \Theta_{1,a} &:= \{(j^*,z^*,\mu_0,\mu_1)\in\Theta: \mu_0-\mu_1\in\{-a,a\}^p, \min(z^*,n-z^*)\geq n\tau\}
\end{align*}
be two subspaces in the parameter space $\Theta$. We consider the problem of testing the null hypothesis $(j^*,z^*,\mu_0,\mu_1)\in\Theta_0$ against the alternative $(j^*,z^*,\mu_0,\mu_1)\in\Theta_1$ using data $X$. 
\begin{theorem}
\label{Thm:LowerBound}
If $n\tau\geq 1$, then for $a^2\leq \frac{\sqrt{\log 2}}{\sqrt{2p}n\tau}$, we have that
 \[
       \inf_{\psi} \Bigl\{\sup_{(j^*,z^*,\mu_0,\mu_1)\in\Theta_0} P_{j^*,z^*,\mu_0,\mu_1}(\psi = 1) + \sup_{(j^*,z^*,\mu_0,\mu_1)\in\Theta_{1,a}} P_{j^*,z^*,\mu_0,\mu_1}(\psi = 0)\Bigr\} \geq 1/2,
\]
where the infimum is taken over all measurable test functions $\psi: \mathbb{R}^{p\times n} \to \{0,1\}$. 
\end{theorem}
In the setting described after Theorem~\ref{thm:testing},  Theorem~\ref{Thm:LowerBound} shows that Condition~\eqref{Eq:CondSignalGeneral} in our estimation result is necessary for the even simpler task of testing the existence of a change-point.

We then consider the special case when we know the sign of the changes in each coordinate. Without loss of generality, we may assume that all changes are positive. In this case, we can use the linear statistics defined in equation~\eqref{Eq:linear} and the following theorem shows that this linear statistic achieves good performance in terms of the estimation consistency:

\begin{theorem}
\label{Thm:LinearStatConsistency}
Suppose $n\tau\geq 2p$ and $X\sim P_{j^*,z^*,\mu_0,\mu_1}$ with $\mu_0-\mu_1\in\{-a\}^p \cup\{a\}^p$. Define $m = m_G(C_1) := \min_{t^*,k^*} |\mathcal{J}_{t^*,k^*}(C_1)|$. There exists a universal constant $c$ such that if $a \geq c\sqrt{\log(pn)/(mn\tau^2)}$, then the estimator $(\hat j, \hat z)$ from~\eqref{Eq:LinearOptimiser} satisfies with probability at least $1-1/(2pn)$ that 
\[
|\hat z - z^*| + d_G(\hat j , j^*) \leq \frac{C^*\log(pn)}{a^2 m}.
\]
\end{theorem}

From this result, we can see that the estimation accuracy of the estimator from the linear statistics achieves the convergence rate of $\log (pn)/ (a^2p)$ and $a^2p$ in the denominator is the $\ell_2$ norm of $\theta$. Similarly rate has also been observed in many change-point results \citep{CM1997}. This condition is also the same as the second term in equation~\ref{Eq:CondSignalGeneral} in Theorem~\ref{Thm:GeneralConsistency}. The following result is an immediate consequence of Theorem~\ref{Thm:LinearStatConsistency} together with Proposition~\ref{Prop:gridcardinality}.
\begin{corollary}
\label{Thm:LinearGridConsistency}
Under the same assumption as in Proposition~\ref{Prop:gridcardinality}. Suppose $X\sim P_{j^*,z^*,\mu_0,\mu_1}$ with $\mu_0-\mu_1\in\{-a\}^p \cup\{a\}^p$. There exist $c, C>0$, depending only on $d$, such that if $a \geq c\sqrt{\log(pn)/(pn\tau^2)}$, then with probability at least $1-1/(2pn)$, the estimator  $(\hat j, \hat z)$ defined in~\eqref{Eq:LinearOptimiser} satisfies that 
\[
|\hat z - z^*| + d_G(\hat j , j^*) \leq \frac{C\log(pn)}{a^2 p}.
\]
\end{corollary}

\section{Numerical studies}
\label{Sec:Simulations}
\subsection{Deterministic spreading model}
\label{Sec:DeterministicSpread}
In this subsection, we compare our method with other possible ways to locate the change. The first possible way is for each row of the data, we perform a one dimension change point testing, that is, pick out the time point with the largest absolute value of the CUSUM statistics for each coordinate. The earliest time and the coordinate corresponding to that time is the desired change point location. We try two different kinds of change-point locations: in the middle and near the end of the boundary. For the first case, We set $n=200$ and vary $p\in\{100,200,500\}$ and signal size $\mu_j^1-\mu_j^0\in\{0.1,0.2,0.5\}$. For the second case, We set $n=500$ and vary $p\in\{500,800,1000\}$ and signal size $\mu_j^1-\mu_j^0\in \{0.2,0.3,0.4,0.5\}$.  We compare the mean absolute deviation between the estimated and true location of $z^*$ and $j^*$ respectively. Columns $\hat z^*_{SD}$ and $\hat j^*_{SD}$ are mean absolute deviation for $z^*$ and $j^*$ from Algorithm~\ref{alg:split} respectively while columns  $\hat z^*_{coordwise}$and $\hat j^*_{coordwise}$ are results from testing procedure stated above. Table~\ref{table:fix} shows that our method can locate the change point accurately especially when $\mu_j^0$,$\mu_j^0$ grows above $0.2$ in both change-point settings.

\begin{table}[tbhp]
\begin{center}
\begin{tabular}{cccccccc}
\hline\hline
$n$ & $p$ &$z^*$ &signal size& $\hat z^*_{SD}$&$\hat z^*_{coordwise}$ & $\hat j^*_{SD}$ &$\hat j^*_{coordwise}$ \\
\hline
$200$ & $100$&  $100$ & $0.1$ & $25.1$ &  $92.24$ &   $20.79$& $46.08$\\
$200$ & $100$&  $100$  & $0.2$ & $2.07$ & $61.44$ &  $2.35$& $33.35$\\
$200$ & $100$ &  $100$  & $0.5$ & $0.06$ & $28.78$ &  $0.07$& $14.91$\\
$200$ & $200$ &  $100$  & $0.1$ & $23.34$ & $85.28$ &  $33.12$& $87.6$\\
$200$ & $200$ &  $100$  & $0.2$ & $1.72$ & $59.36$ & $1.69$ & $62.19$\\
$200$ & $200$ &  $100$  & $0.5$ & $0.01$ & $29.78$  & $0.01$& $19.92$\\
$200$ & $500$ &  $100$  & $0.1$ & $59.84$ &$87.24$  & $77.93$ & $204.56$\\
$200$ & $500$ &  $100$  & $0.2$ & $4.14$ &$60.92$   &$4.05$ & $110.07$\\
$200$ & $500$ &  $100$  & $0.5$ & $0$ &  $34.61$ & $0$& $26.35$\\
\hline
$500$ & $500$&  $400$   & $0.2$ & $10.4$ & $106.24$ & $10.36$ & $101.2$\\
$500$ & $500$ &  $400$  & $0.3$ & $0.2$ & $98.02$ & $0.16$ & $31.39$\\
$500$ & $500$ &  $400$  & $0.4$ & $0.04$ & $121.48$ & $0.03$ & $30.65$\\
$500$ & $500$ &  $400$  & $0.5$ & $0$ & $131.16$ & $0$ & $30.18$\\
$500$ & $800$&  $400$   & $0.2$ & $51.59$ & $161.95$ & $51.9 $ & $142.05$\\
$500$ & $800$ &  $400$  & $0.3$ & $0.18$ & $160.9$ & $0.15$ & $97.21$\\
$500$ & $800$ &  $400$  & $0.4$ & $0$ & $179.76$ & $0$ & $86.56$\\
$500$ & $800$ &  $400$  & $0.5$ & $0$ & $171.38$ & $0$ & $87.75$\\
$500$ & $1000$&  $400$   & $0.2$ & $77.04$ & $160.7$ & $77.18$ & $173.15$\\
$500$ & $1000$ &  $400$  & $0.3$ & $0.22$ & $158.25$ & $0.2$ & $104.48$\\
$500$ & $1000$ &  $400$  & $0.4$ & $0.02$ & $166.98$ & $0$ & $98.88$\\
$500$ & $1000$ &  $400$  & $0.5$ & $0.01$ & $171.3$ & $0$ & $94.17$\\
\hline\hline
\end{tabular}
\caption{\label{table:fix}Average mean absolute deviation (over 100 repetitions) comparison between different methods. Other parameters used: $j^*=p/2$.}
\end{center}
\end{table}

\subsection{Stochastic spreading model}
\label{Sec:StochasticSpread}
In this subsection, we consider the case when the spread of the change occur independently with probability $q$ each time from an infected node to each of its neighbours. In this case, we can modify our existing methodology, which monitors for deterministic spreading of the change as follows. if the probability $q$ is known, then we can adjust the distance between coordinates $j$ and $k$ as the expected time that a change spreading from source coordinate $j$ will reach $k$ under this stochastic model. For a line graph $G = C_p$, this would simply be $d_G(j,k) / q$. When $q$ is unknown, we may search over a grid $\mathcal{Q}$ of $q$ values in $[0,1]$, compute the test statistics $\max_{j,t} Q^{(q)}_{j,t}$ for each $q$ as in~\eqref{Eq:QuadraticEstimator} with this adjusted distance metric and then choose the optimal $q$ by $\hat q := \argmax_{q\in\mathcal{Q}} \max_{j,t} Q^{(q)}_{j,t}$. The final estimator for the source coordinate and the time of change-point is defined as $(\hat j, \hat t):=\argmax_{(j,t)} Q^{(\hat q)}_{j,t}$.  In Table~\ref{table:compare}, we compare the performance of the method described above (denoted by rSD) and the vanilla SpreadDetect algorithm (denoted by SD), together with the baseline coordinatewise procedure mentioned in Section~\ref{Sec:DeterministicSpread}. 
We set the true probability of change spread to $q=0.5$, and search over the grid $\mathcal{Q} = \{0.1,0.2,\ldots,0.9\}$ and vary $n$, $p$, $z^*$ and $j^*$. We see that the modified SpreadDetect algorithm described in this subsection has the best performance over the wide range of parameter settings considered.

\begin{table}[tbhp]
\begin{center}
\begin{tabular}{ccccccccccc}
n & p &$z^*$ & $j^*$ &signals& $\hat z^*_{SD}$& $\hat z^*_{rSD}$ &$\hat z^*_{coordwise}$ & $\hat j^*_{SD}$ &$\hat j^*_{rSD}$ &$\hat j^*_{coordwise}$ \\
\hline\hline
$200$ & $100$ & $100$ & $50$ & $0.2$ & 17.29 & $9.05$ & $97.79$ & $6.64$ & $3.35$ & $48.59$\\
$200$ & $100$ & $100$ & $50$ & $0.3$ & 27.86 & $4.35$ & $83.97$ & $3.41$ & $1.89$ & $41.79$\\
$200$ & $100$ & $100$ & $50$ & $0.4$ & 16.67 & $3.57$ & $41.47$ & $2.53$ & $1.62$ & $21.46$\\
$200$ & $200$ & $100$ & $100$ & $0.2$ & 19.23 & $20.67$ & $96.43$ & $23.07$ & $16.33$ & $96.29$\\
$200$ & $200$ & $100$ & $100$ & $0.3$ & 19.10 & $5.79$ & $88.02$ & $11.24$ & $2.38$ & $85.25$\\
$200$ & $200$ & $100$ & $100$ & $0.4$ & 17.07 & $3.66$ & $52.35$ & $4.6$ & $1.81$ & $46.18$\\
$200$ & $500$ & $100$ & $250$ & $0.2$ & 61.09 & $44.43$ & $98.15$ & $77.36$ & $42.95$ & $246$\\
$200$ & $500$ & $100$ & $250$ & $0.3$ & 39.23 & $7.82$ & $87.6$ & $42.23$ & $5.3$ & $209.06$\\
$200$ & $500$ & $100$ & $250$ & $0.4$ & 22.59 & $4.16$ & $46.55$ & $10.6$ & $1.71$ & $83.5$\\
$500$ & $200$ & $250$ & $100$ & $0.2$ & 41.69 & $6.68$ & $152.61$ & $5$ & $2.61$ & $64.23$\\
$500$ & $200$ & $250$ & $100$ & $0.3$ & 41.47 & $5.77$ & $29.53$ & $3.92$ & $2.36$ & $14.77$\\
$500$ & $200$ & $250$ & $100$ & $0.4$ & 41.15 & $5.39$ & $43.72$ & $3.37$ & $2.39$ & $12.28$\\
$500$ & $500$ & $250$ & $250$ & $0.2$ & 43.69 & $5.76$ & $170.54$ & $7.64$ & $2.46$ & $151.04$\\
$500$ & $500$ & $250$ & $250$ & $0.3$ & 42.34 & $5.09$ & $35.02$ & $6.16$ & $2.38$ & $14.48$\\
$500$ & $500$ & $250$ & $250$ & $0.4$ & 43.02 & $5.19$ & $44.59$ & $4.99$ & $2.41$ & $13.11$\\
\hline\hline
\end{tabular}
\caption{\label{table:compare}Average mean absolute deviation (over 100 repetitions) comparison between different methods for estimating the time of change-point and source coordinate under a stochastic spreading model described in Section~\ref{Sec:StochasticSpread}}
\end{center}
\end{table}

\subsection{Real data example}
We now apply Algorithm~\ref{alg:split} to the data set of weekly death between January 2017 and December 2020 in United States. The aim is to find the time of the change in number of deaths and state where the change first occurs. We exclude two states: Alaska and Hawaii in our analysis as they have no adjacent states. To form the adjacency matrix, if two states are adjacent to each other, then we assign the corresponding entry with $1$, otherwise, the entries are $0$. Before applying Algorithm~\ref{alg:split} to the data, we first remove the seasonal trend from the data. Specifically, we use the data up to 30 June 2019 as the training data and estimate the daily death by averaging the weekly total death and then use a Gaussian Kernel with bin width of $20$ to estimate the deaths  on each day of a year. As daily death follows Poisson distribution, we stabilize the variance by applying a square root transformation. Then we calculate the difference between actual data with the fitted data and standardize it using the mean and standard deviation of the calculated difference.

We apply Algorithm~\ref{alg:split} to the pre-processed data set. The resulting time is 7 March 2020, and the state which first start to change is Pennsylvania. The date matches the actual situation, as during that time, death due to COVID 19 began to occur. Figure~\ref{fig:real} shows the aggregated CUSUM statistics with the states arranged such that Pennsylvania is in the centre and the graph distance increases as we move towards top and bottom of the plot. The heatmap shown in the figure is consistent with a change spreading from Pennsylvania. However, we remark that the conclusion here should be treated with caution for two reasons. Firstly, this is a weekly recorded data and the frequency of recordinly is likely to be inadequate to capture the rapid spreading of the disease across multiple states. Secondly, we computed the distance between states by the number of state boarders one needs to cross from one to the other. While this is a proxy for the distance between states during the pandemic spread, a better measure would involve for instance the number of passengers crossing from one state to another, though the latter data are difficult to obtain. 

\begin{figure}[tbp]
\centering
\includegraphics[width=11cm]{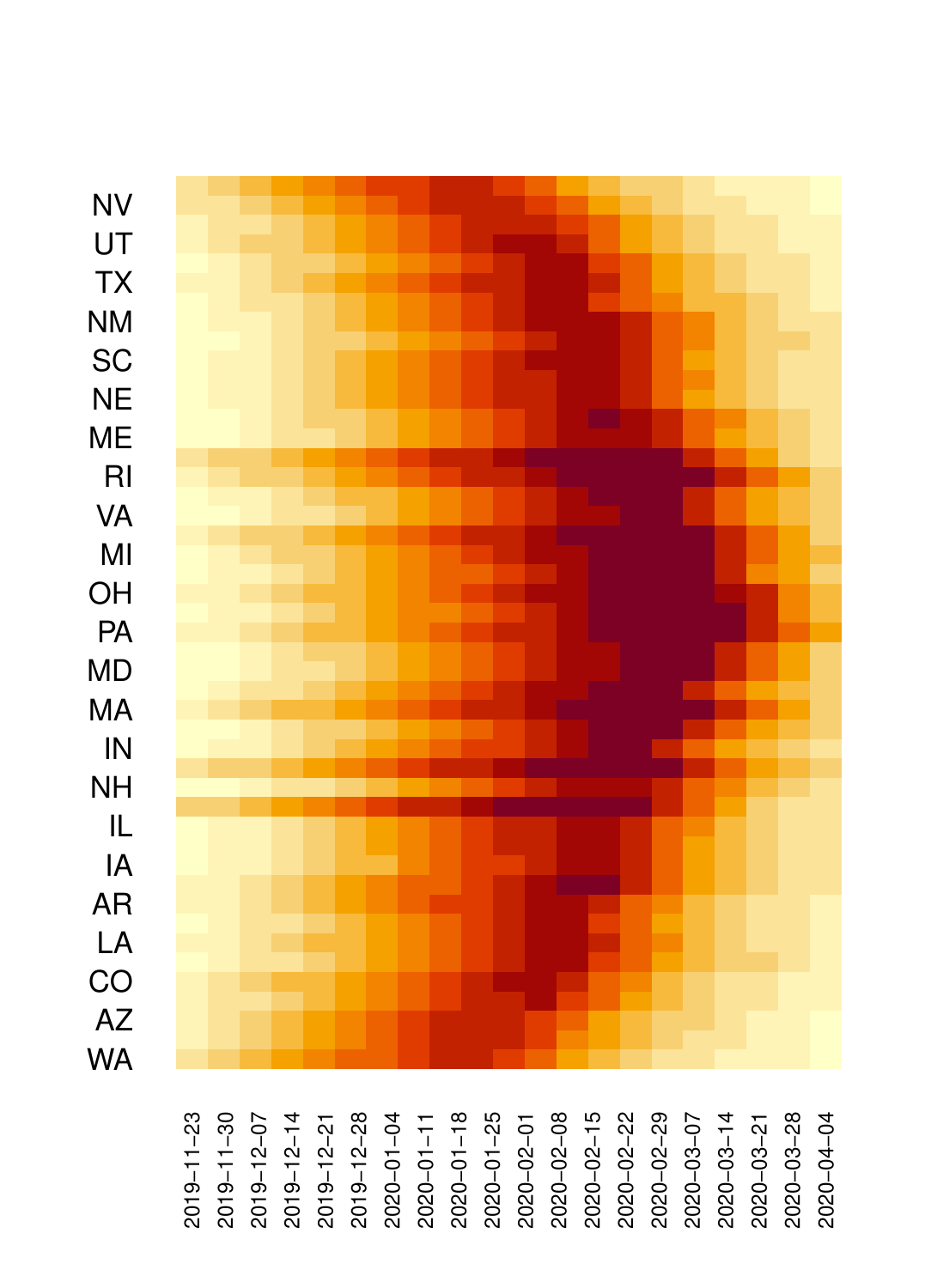} 
\caption{\label{fig:real}Aggregated CUSUM statistics for 46 states during 2019-11-23 to 2020-04-04.}
\end{figure}

\appendix
\section{Proof of main results}
\label{Sec:Proofs}
\begin{proof}[Proof of Theorem~\ref{Thm:GeneralConsistency}]
Let $A_{j,t}$ be the entries of $A = \mathcal{T}(\boldsymbol{\mu})$ then for $j\in [p]$, we have 
\[
A_{j,t}=
\begin{cases}
  \sqrt{\frac{t}{n(n-t)}}(n-z^*-d_G(j,j^*))\theta_j, & \text{ if $t\leq z^*+d_G(j,j^*)$}, \\ \sqrt{\frac{n-t}{nt}}(z^*+d_G(j,j^*))\theta_j, &\text{ if $t>z^*+d_G(j,j^*)$}.
  \end{cases}
\]
 Since the test statistic is unchanged by flipping signs in any one row of the data, we may assume without loss of generality that $\theta_j>0$ for all $j$. 
 
Fix $k^* \in[p]$ and $t^* \in[n-1]$, we definte 
\[
B_{k^*,t^*}:= \sum_{ j\in\mathcal{J}_{k^*,t^*}} A_{j, t^*+d_G(k^*,j)}^2.
\]
For each $j$ such that $t^*+d_G(k^*,j) < n$, we have $T_{j,t^*+d_G(k^*,j)}\sim N(A_{j,t^*+d_G(k^*,j)},1)$ and obtain that $Q_{k^*, t^*} + |\mathcal{J}_{k^*,t^*}| \sim \chi^2_{\mathcal{J}_{k^*,t^*}} (B_{k^*,t^*})$. Therefore, by \citet[Lemma~8.1]{birge2001alternative}, for each $j\in[p]$ and $t\in[n-1]$, we have for any $\delta \in (0,1)$ that 
\[
\mathbb{P}\Bigl( |Q_{k^*,t^*} - B_{k^*,t^*}| >   2\sqrt{(|\mathcal{J}_{k^*,t^*}|+2B_{k^*,t^*})\log(2/\delta)}+2\log(2/\delta)\Bigr)\leq \delta.
\]
Taking a union bound over $k^*\in[p]$ and $t^*\in[n-1]$ of the above inequality, we therefore obtain that with probability at least $1-\delta$,
\begin{align}
B_{j^*, z^*} &-2\sqrt{(|\mathcal{J}_{j^*,z^*}|+2B_{j^*,z^*})\log(2pn/\delta)}-2\log(2pn/\delta)\leq Q_{j^*,z^*} \leq Q_{\hat j, \hat z} \nonumber\\
&\hspace{2cm}\leq  B_{\hat j, \hat z} + 2\sqrt{(|\mathcal{J}_{\hat j,\hat t}|+2B_{\hat j,\hat t})\log(2pn/\delta)}+2\log(2pn/\delta).\label{Eq:Thm1tmp1}
\end{align}
Notice that for every $k^*\in[p]$ and $t^*\in[n-1]$, we have $|\mathcal{J}_{k^*,t^*}|\leq p$ and 
\[
B_{k^*,t^*} \leq \sum_{j\in[p]} A_{j,z^*+d_G(j,j^*)}^2\leq \sum_j \theta_j^2 \frac{(z^*+d_G(j,j^*))(n-z^*-d_G(j,j^*))}{n} \leq \frac{n\|\theta\|_2^2}{2}.
\]
Thus, after rearranging~\eqref{Eq:Thm1tmp1}, we have with probability at least $1-\delta$ that 
\begin{equation}
\label{Eq:BjUpper}
B_{j^*, z^*} - B_{\hat j, \hat z} \leq 4\sqrt{(p+n\|\theta\|_2^2)\log(2pn/\delta)}+4\log(2pn/\delta).
\end{equation}

On the other hand, we can obtain a lower bound of $B_{j^*,z^*}-B_{\hat j, \hat z}$ as follows. For each $j\in [p]$, the sequence $(A_{j,t})_t$ is unimodal with a single peak at $z^*+d_G(j,j^*)$. Moreover, since $d_G(j,j^*) \leq p \leq n\tau/2$, we have
\begin{equation}
\label{Eq:Apeak}
\theta_j A_{j,z^*+d_G(j,j^*)}=\theta_j^2\sqrt{\frac{(z^* + d_G(j,j^*) (n- z^* - d_G(j,j^*))}{n}}\geq \theta_j^2\sqrt{\frac{(n\tau/2)(n/2)}{n}}\geq \frac{\theta_j^2\sqrt{n\tau}}{2}.
\end{equation}
Therefore, by \citet[Lemma~7]{wang2016highdimensional}, we have for  each $j\in \mathcal{J}_{\hat z,\hat j}(C_1)$ that 
\begin{align}
\label{Eq:Aj_diffGeneral}
\theta_j (A_{j,z^*+d_G(j,j^*)} - A_{j,\hat{z}+d_G(j,\hat j)}) &\geq \frac{2\theta_j^2}{3\sqrt{6 n\tau}}\min\biggl\{|z^*+d_G(j,j^*)-\hat{z}-d_G(j,\hat j)|, \frac{n\tau}{2} \biggr\}\nonumber\\
 & \geq \frac{2\theta_j^2}{3\sqrt{6 n\tau}}\min\biggl\{C_1 (|z^*-\hat z|+d_G(j^*,\hat j)), \frac{n\tau}{2}\biggr\},
\end{align}
where we have used the definition of $\mathcal{J}_{\hat z, \hat j}(C_1)$ from~\eqref{Eq:mathcalJ} in the final bound. Combining~\eqref{Eq:Aj_diffGeneral} with~\eqref{Eq:Apeak}, we obtain that 
\begin{align}
B_{j^*,z^*}-B_{\hat j,\hat z} &\geq \sum_{j\in \mathcal{J}_{\hat z, \hat j}(C_1)} (A_{j,z^*+d_G(j,j^*)}^2 - A_{j, \hat z+d_G(j,\hat j)}^2) \nonumber\\
&\geq \sum_{j\in \mathcal{J}_{\hat z, \hat j}(C_1)} A_{j,z^*+d_G(j,j^*)}(A_{j,z^*+d_G(j,j^*)}-A_{j, \hat z+d_G(j,\hat j)})\nonumber\\
 &\geq \frac{2a^2 m}{3\sqrt{6}} \min\biggl\{C_1 (|z^*-\hat z|+d_G(j^*,\hat j)), \frac{n\tau}{2}\biggr\},\label{Eq:BjLower}
\end{align}
where we have used the fact that $|\mathcal{J}_{\hat z, \hat j}(C_1)| \geq m$ in the final inequality. 
Combining~\eqref{Eq:BjUpper} and~\eqref{Eq:BjLower}, and choosing $\delta = 1/(2pn)$, we have with probability at least $1-1/(2pn)$ that 
\begin{equation}
\label{Eq:BasicInequalityGeneral}
\frac{2a^2m}{3\sqrt{6}} \min\biggl\{C_1 (|z^*-\hat z|+d_G(j^*,\hat j)), \frac{n\tau}{2}\biggr\}\leq 4\sqrt{p+2npa^2\log(2pn)}+8\log(2pn).
\end{equation}
From condition~\eqref{Eq:CondSignalGeneral}, we have
\[
a^2 \geq \max\biggl\{\frac{c(\sqrt{p}+\log(2pn))}{n\tau m}, a\sqrt{\frac{cp\log(2pn)}{n\tau^2m^2}} \biggr\} \geq \frac{c\sqrt{p} + c\log(2pn) + \sqrt{cnpa^2\log(2pn)}}{2n\tau m},
\]
which for sufficiently large $c$ implies that the minimum on the left-hand side of~\eqref{Eq:BasicInequalityGeneral} is achieved by the first term. Consequently, we derive from~\eqref{Eq:BasicInequalityGeneral} that with probability at least $1-1/(2pn)$,
\[
|z^*-\hat z| + d_G(j^*-\hat j) \leq \frac{12\sqrt{6}}{C_1}   \biggl\{ \frac{\sqrt{p}+\log(2pn)}{a^2 m} + \frac{\sqrt{pn \log(2pn)}}{a m}\biggr\},
\]
as desired.
\end{proof}

\begin{proof}[Proof of Proposition~\ref{Prop:gridcardinality}]
Denote $G_r$ to be the $r$th copy of $C_{p_1}$ making up $G$, i.e.\ $G = \prod_{r=1}^d G_r$. Each vertex $j \in G$ can be represented by a $d$-tuple of coordinates $(\pi_1(j), \ldots, \pi_d(j))$, where $\pi_r(j) \in V(G_r) = [p_1]$. Define $\ell(j) = \ell_G(j) := (z^*+d_G(j,j^*)) - (t^* + d_G(j, k^*))$, by Proposition~\ref{Prop:SetJ1}, we have that each of the following set 
\begin{align*}
  \mathcal{J}_r := \biggl\{\tilde j: \mathrm{sgn}(z^*-t^*)\ell_{G_r}(\tilde j)\geq \frac{|z^*-t^*|+d_{G_r}(\pi_r(j^*),\pi_r(k^*))}{4}\biggr\},  
\end{align*}
has cardinality at least $p_1/8$. Then, for all $j\in \mathcal{J}:=\prod_{r=1}^d \mathcal{J}_r$, when $z^*\geq t^*$, we have:
\begin{align*}
    z^*-t^*+d_G(j,j^*)-d_G(j,k^*)&= \sum_{r=1}^d \biggl\{\frac{z^*-t^*}{d} + d_{G_r}(\pi_r(j),\pi_r(j^*))-d_{G_r}(\pi_r(j),\pi_r(k^*))\biggr\}\\
    &\geq  \sum_{r=1}^d \frac{\ell_{G_r}(j)}{d}\geq  \sum_{r=1}^d \frac{z^*-t^*+d_{G_r}(\pi_r(j^*),\pi_r(k^*))}{4d}\\
    &\geq \frac{z^*-t^*}{4d}+\frac{d_G(j^*,k^*)}{4d}
\end{align*}
Similarly, if $t^*>z^*$, we have 
\[
t^*-z^* + d_G(j,k^*) - d_G(j,j^*) \geq \frac{t^*-z^*}{4d} + \frac{d_G(j^*,k^*)}{4d}.
\]
Overall, we have for $j\in\mathcal{J}$ that 
\[
|\ell_G(j)| = |z^*-t^* + d_G(j,j^*)-d_G(j,k^*)| \geq \frac{|z^*-t^*|}{4d} + \frac{d_G(j^*,k^*)}{4d}.
\]
Hence, $m_G(1/(4d)) \geq |\mathcal{J}| = \prod_{r=1}^d|\mathcal{J}_r|\geq (p_1/8)^d = p/8^d$ as desired.
\end{proof}

\begin{proof}[Proof of Theorem~\ref{thm:testing}]
    If $\theta=0$, then $T_{j,t}\sim N(0,1)$ for all $t\in[n-1]$ and $j\in[p]$ and hence $Q_{j,t} + |\mathcal{J}_{j,t}| \sim \chi^2_{|\mathcal{J}_{j,t}|}$. By \citet[Lemma 1]{LM2000} together with a union bound, we have that 
    \begin{align*}
        \mathbb{P}(\max_{j\in[p],t\in[n-1]} Q_{j,t} \geq \lambda)&\leq\sum_{j=1}^p\sum_{t=1}^{n-1} \mathbb{P}(Q_{j,t}\geq \lambda)\\
        &\leq \sum_{j=1}^p\sum_{t=1}^{n-1} \mathbb{P}\bigl\{Q_{j,t}\geq 2\sqrt{|\mathcal{J}_{j,t}|\log(pn/\delta)}+2\log(pn/\delta)\bigr\}\leq \delta.
    \end{align*}
    This establishes part (a). For part (b), let $A_{j,t}$ and $B_{j,t}$ be defined as in the proof of Theorem~\ref{Thm:GeneralConsistency}. Note that under the alternative hypothesis, $Q_{j,t} + |\mathcal{J}_{j,t}| \sim \chi^2_{|\mathcal{J}_{j,t}|}(B_{j,t})$. Hence, by \citet[Lemma~8.1]{birge2001alternative}, we have
    \[
    \mathbb{P}\Bigl\{Q_{j^*,z^*} \geq B_{j^*,z^*} - 2\sqrt{(|\mathcal{J}_{j^*,t^*}| + 2B_{j^*,z^*})\log(1/\delta)}\Bigr\} \geq 1- \delta. 
    \]
  Under the assumption that $B_{j^*,z^*} \geq 8\lambda$, we have
  \begin{align*}
  B_{j^*,z^*} - 2\sqrt{(|\mathcal{J}_{j^*,t^*}| + 2B_{j^*,z^*})\log(1/\delta)} &\geq   B_{j^*,z^*} - 2\sqrt{p\log(1/\delta)} - 2\sqrt{2 B_{j^*,z^*}\log(1/\delta)}\\
  & \geq B_{j^*,z^*}  - \lambda - 2\sqrt{B_{j^*,z^*} \lambda}\\
  & = (\sqrt{B_{j^*,z^*}} - \sqrt{\lambda})^2 - 2\lambda \geq  \lambda.
  \end{align*}

 Since 
 \[
A_{j,z^*+d_G(j,j^*)}^2=\theta_j^2\frac{(z^* + d_G(j,j^*) (n- z^* - d_G(j,j^*))}{n}\geq \theta_j^2\frac{(n\tau/2)(n/2)}{n}\geq \frac{\theta_j^2n\tau}{4},
 \]
 there are at least $n\tau/2$ points with $d_G(j,j^*)\leq n\tau/2$,
 \[
 B_{j^*,z^*}=\sum_{j\in\mathcal{J}_{j^*,z^*}} A_{j, z^*+d_G(j^*,j)}^2\geq \min\bigl(p,\frac{n\tau}{2}\bigr)\frac{a^2n\tau}{4}=\frac{a^2n\tau\min(2p,n\tau)}{8 }
 \]
 Then, for $a^2\geq 64\lambda/(n\tau\min\{2p,n\tau\})$, we have the desired result.
\end{proof}

\begin{proof}[Proof of Theorem~\ref{Thm:LowerBound}]
Fix $j^*\in \argmin_{v\in V(G)} \max_{w\in V(G)} d_G(v,w)$ and $z^* = n - \lceil n\tau\rceil$.  Let $\pi$ be the uniform distribution on $\{-a,a\}^p$. For notational simplicity, define $P_0 := P_{j^*,z^*,0,0}$ and $P_1:= \int P_{j^*,z^*,0,\mu_1} d\pi(\mu_1)$. Then, for any test function $\psi$, we have 
\begin{align}
 \sup_{(j^*,z^*,\mu_0,\mu_1)\in\Theta_0} & P_{j^*,z^*,\mu_0,\mu_1}(\psi = 1) + \sup_{(j^*,z^*,\mu_0,\mu_1)\in\Theta_1} P_{j^*,z^*,\mu_0,\mu_1}(\psi = 0)\nonumber\\ 
 & \geq P_0(\psi=1) + P_1(\psi=0) \geq 1 - d_{\mathrm{TV}}(P_0, P_1)\nonumber\\
 & = 1 - \frac{1}{2}\int \biggl|\frac{dP_1}{dP_0} - 1\biggr|\,dP_0 \geq 1 - \frac{1}{2}\biggl\{\int \biggl(\frac{dP_1}{dP_0} - 1\biggr)^2\,dP_0\biggr\}^{1/2}\nonumber\\
 & = 1 - \frac{1}{2}\biggl\{\int \biggl(\frac{dP_1}{dP_0}\biggr)^2 \,dP_0 - 1\biggr\}^{1/2}.\label{Eq:ChiSquaredDivergence}
\end{align}
Let $\boldsymbol{\mu}$ be the conditional mean of $X$ given $\mu_1$ under $P_{j^*,z^*,0,\mu_1}$ and let $\tilde{\boldsymbol{\mu}}$ be an independent copy of $\boldsymbol{\mu}$.  By \citet{Ingster}, we have for some independent Rademacher random variables $\xi_1,\ldots,\xi_p$ that
\begin{align*}
\int \biggl(\frac{dP_1}{dP_0}\biggr)^2 \,dP_0 &= 
\mathbbm{E} \exp \langle \mu, \tilde\mu\rangle  =\mathbbm{E}\exp\bigg(\sum_{j=1}^p \max\{n-z^*-d_G(j,j^*), 0\}a^2 \xi_j\bigg)\\
&=\prod_{j=1}^p \Bigl[\frac{1}{2} e^{\max\{n-z^*-d_G(j,j^*),0\}a^2} + \frac{1}{2} e^{-\max\{n-z^*-d_G(j,j^*),0\}a^2}\Bigr]\\
 &\leq\prod_{j=1}^p e^{\max\{n-z^*-d_G(j,j^*),0\}^2a^4/2} \leq e^{2pn^2\tau^2a^4} \leq 2,
\end{align*}
where the first inequality follows from the fact that $(e^x+e^{-x})/2\leq e^{x^2/2}$ for all $x\in \mathbbm{R}$ and the second bound uses the fact that $n-z^*-d_G(j,j^*) \leq \lceil n\tau\rceil \leq 2n\tau$. Finally, substituting the above inequality into~\eqref{Eq:ChiSquaredDivergence} we arrive at the desired conclusion.
\end{proof}

\begin{proof}[Proof of Theorem~\ref{Thm:LinearStatConsistency}]
From the definition of $(\hat j, \hat z)$, we have $\sum_{j\in\mathcal{J}_{\hat z, \hat j}(C_1)} (A_{j,z^*+d_G(j,j^*)}+E_{j,z^*+d_G(j,j^*)})\leq \sum_{j\in\mathcal{J}_{\hat z, \hat j}(C_1)} (A_{j,\hat z+d_G(j,\hat j)}+E_{\hat z+d_G(j,\hat j)})$, which can be combined with Proposition~\ref{Prop:Variance} to obtain that for some universal constant $C_2 > 0$, we have with probability at least $1-1/(pn)$ that 
\begin{equation}
\label{Eq:A_gap_upper}
    \sum_{j\in\mathcal{J}_{\hat z, \hat j}(C_1)} (A_{j,z^*+d_G(j,j^*)}-A_{j,\hat z +d_G(j,\hat j)})\leq C_2 \biggl\{|\mathcal{J}_{\hat z, \hat j}(C_1)|\log (pn) \frac{|z^*-\hat z|+d_G(j^*,\hat j)}{n\tau}\biggr\}^{1/2}.
\end{equation}
On the other hand, by~\eqref{Eq:Aj_diffGeneral}, we have  
\begin{align}
 \sum_{j\in\mathcal{J}_{\hat z, \hat j}(C_1)} (A_{j,z^*+d_G(j,j^*)}-A_{j,\hat z +d_G(j,\hat j)}) \geq \frac{2a|\mathcal{J}_{\hat z, \hat j}(C_1)|}{3\sqrt{6 n\tau}}\min\biggl\{C_1 (|z^*-\hat z|+d_G(j^*,\hat j)), \frac{n\tau}{2}\biggr\}.\label{Eq:LinearLower}
\end{align}
Combining~\eqref{Eq:A_gap_upper} and~\eqref{Eq:LinearLower}, we have with probability at least $1-1/(pn)$ that 
\begin{equation}
\label{Eq:BasicInequalityLinear}
\frac{2a|\mathcal{J}_{\hat z, \hat j}(C_1)|^{1/2}}{3\sqrt{6\log(pn)}} \min\biggl\{C_1 (|z^*-\hat z|+d_G(j^*,\hat j)), \frac{n\tau}{2}\biggr\}\leq  C_2 \bigl\{|z^*-\hat z|+d_G(j^*,\hat j)\bigr\}^{1/2} .
\end{equation}
We claim that when $c\geq 6\sqrt{3}C_2$, the minimum on the left-hand side above cannot
be achieved at $\frac{n\tau}{2}$. Indeed, from the assumption on $a$,
we have
\begin{align*}
\frac{2a|\mathcal{J}_{\hat z, \hat j}(C_1)|^{1/2}}{3\sqrt{6\log(pn)}}\frac{n\tau}{2} &\geq \frac{c\sqrt{n}}{3\sqrt{6}} \geq C_2\sqrt{2n} > C_2 \bigl\{|z^*-\hat z|+d_G(j^*,\hat j)\bigr\}^{1/2}.
\end{align*}
Therefore, we have that 
\[
|z^*-\hat z|+d_G(j^*,\hat j)\leq \frac{27C_2^2 \log (pn)}{2C_1^2 a^2 |\mathcal{J}_{\hat z, \hat j}(C_1)|}\leq \frac{C^*\log(pn)}{a^2m}.
\]
\end{proof}

\section{Ancillary results}
\label{Sec:Ancillary}
We first show that when $G = C_p$, a cycle graph, for a nontrivial fraction of coordinates $j\in[p]$, the difference $\ell(j) = \ell_G(j) := (z^*+d_G(j,j^*)) - (t^* + d_G(j, k^*))$ is large in absolute value. 

\begin{prop}
\label{Prop:SetJ}
Let $G = C_p$ be a $p$-cycle graph. Let $\tau = \min\{z^*/n, 1-z^*/n\}$. Assuming that $n\tau\geq 16$ and $p\geq 4$, the following set
\[
   \mathcal{J} := \biggl\{j: |\ell(j)|\geq \min\biggl(\frac{n\tau}{16}, \frac{|z^*-t^*|+d_G(j^*,k^*)}{4}\biggr)\text{ and } d_G(j , j^*) \leq \frac{n\tau}{2} \biggr\}           
\]
has cardinality at least $\min(p, n\tau)/32$.
\end{prop}

\begin{proof}
Without loss of generality, we may assume by symmetry that $j^* = \lceil p/2 \rceil$ and $k^* \geq j^*$. This choice is convenience since $d_G(j, j^*) = |j - j^*|$. With this choice, we can write
\begin{equation}
\label{Eq:ell}
 \ell(j) = \begin{cases}
            (z^*-t^*) + (k^*-j^*) -2j & 1\leq j \leq k^*-j^*\\
            (z^*-t^*) - (k^*-j^*) & k^* - j^* \leq j \leq j^*\\
            (z^*-t^*)+ (k^*-j^*) - 2(k^*-j) & j^* \leq j\leq k^*\\
            (z^*-t^*)+ (k^*-j^*) & k^* \leq j \leq p.
           \end{cases}
\end{equation}

We then prove the result by discussing the following four cases.

\textbf{Case 1}: assume $k^*-t^* \geq j^*-z^*$ and $k^*+t^* \leq j^*+z^*$. In this case, we have $t^*\leq z^*$ and $k^*-j^*\leq z^*-t^*$. Hence $\ell(j) \geq 0$ for all $j$.  Notice that $\ell(j)$ is an non-decreasing function of $j$ for $j\geq j^*$. Then for all $j$ such that 
\[
j^* + \min\{n\tau/4, (k^*-j^*)/4\} \leq j\leq j^* + \min\{n\tau/2, p/4\},
\]
we have 
\begin{align*}
\ell(j) & \geq \min\bigl\{\ell(j^* + \lceil n\tau/4\rceil), \ell(j^*+\lceil (k^*-j^*)/4\rceil)\bigr\} \\
&\geq \min\biggl\{(z^*-t^*) + \min\biggl\{\frac{n\tau}{2} - (k^*-j^*), k^*-j^*\biggr\}, z^*-t^*-\frac{1}{2}(k^*-j^*)\biggr\} \\
&\geq \min\biggl\{ \frac{n\tau}{2}, \frac{|z^*-t^*|+|j^*-k^*|}{4}\biggr\}.
\end{align*}
 Consequently, in this case, we have 
\begin{align*}
|\mathcal{J}| \geq \min\{\lfloor n\tau/4 \rfloor , p/8\}.
\end{align*}

\textbf{Case 2}: assume $k^*-t^* \geq j^*-z^*$, $k^*+t^* \geq j^*+z^*$ and $z^*\geq t^*$. In this case, $k^*-j^*\geq z^*-t^*\geq 0$.
We define $h^*$ to be the point such that $h^*-j^*=\lceil \frac{(k^*-j^*)-(z^*-t^*)}{2}\rceil$. Then $j^*\leq h^*\leq k^*$ and $\ell(h^*) \in\{0,1\}$. 

We discuss three sub-cases. \underline{Case 2a}: When $k^*-j^* \leq n\tau/4$, let $A=\{j:  \frac{k^*+h^*}{2} \leq j\leq \min(j^* + \frac{n\tau}{2}, p)\}$.
Then, for all $j\in A$, we have that 
\begin{align*}
\ell(j) &\geq \ell\biggl(\biggl\lceil \frac{h^*+k^*}{2}\biggl\rceil\biggr)\geq \ell(h^*) + 2\biggl\lceil \frac{k^*-h^*}{2}\biggr\rceil \\
&\geq \ell(h^*) +  \biggl\lfloor \frac{(z^*-t^*)+(k^*-j^*)}{2}\biggr\rfloor \geq \frac{(z^*-t^*)+(k^*-j^*)}{4}.
\end{align*}
and 
\begin{align*}
    |A|=\min\biggl(p,j^*+\frac{n\tau}{2}\biggr)-\frac{k^*+h^*}{2}&=\min\biggl\{p-\frac{k^*+h^*}{2},\frac{n\tau}{2}-\biggl(\frac{k^*+h^*}{2}-j^*\biggr)\biggr\}\\
    &\geq\min\biggl(p-k^* + \frac{k^*-h^*}{2}, n\tau/4 \biggr) \geq\min\biggl(\frac{p-j^*}{8},  n\tau/4 \biggr)\\
    &\geq \min(p/32, n\tau/4).
\end{align*}

\underline{Case 2b}: When $h^*-j^*\geq n\tau/8$, let  $A=\{j: j^*\leq j\leq\frac{j^*+h^*}{2}\}$.
Then, 
\[
|\ell(j)|\geq \biggl|\ell\biggl(\biggl\lfloor\frac{j^*+h^*}{2}\biggr\rfloor\biggr)\biggr| \geq 2 \biggl\{h^*-\biggl\lfloor\frac{j^*+h^*}{2}\biggr\rfloor\biggr\} - \ell(h^*) \geq \frac{n\tau}{16}, 
\]
and $|A|\geq \lfloor (h^*-j^*)/2\rfloor \geq \lfloor n\tau/16 \rfloor\geq n\tau/32$. 

\underline{Case 2c}: When $k^*-j^*>n\tau/4$ and $h^*-j^*<n\tau/8$, let  $A=\{j: h^*+\frac{n\tau}{16}\leq j^* \leq h^*+\frac{n\tau}{8}\}$.
Then,
\[
|\ell(j)|\geq\ell(h^*)+2\biggl\lceil \frac{n\tau}{16}\biggl\rceil\geq \frac{n\tau}{8},
\]
and $|A|\geq \lfloor n\tau/16 \rfloor\geq n\tau/32$.

Combining all subcases and noticing that $A\subseteq \mathcal{J}$, we have the desired result.

\textbf{Case 3}: assume $k^*-t^* \geq j^*-z^*$, $k^*+t^* \geq j^*+z^*$ and $t^* > z^*$. We define $h^*$ to be the point such that $h^*-j^*=\lfloor \frac{(k^*-j^*)+(t^*-z^*)}{2}\rfloor$. Observe that $\ell(h) \in\{0,-1\}$. In this case, $k^*-j^*\geq t^*-z^*\geq 0$. Let  $A=\{j: j^*-\min\{2j^*-k^*, n\tau/2\}\leq j\leq j^*+\min\{\frac{h^*-j^*}{2}, n\tau/2\}$. Since $\ell(j)$ is negative and increasing for $j\in [k^*-j^*, \frac{h^*+j^*}{2}]$, we have for all $j\in A$ that 
\[
|\ell(j)|\geq \biggl|\ell\biggl(\biggl\lfloor\frac{j^*+h^*}{2}\biggr\rfloor\biggr)\biggr| \geq 2 \biggl\{h^*-\biggl\lfloor\frac{j^*+h^*}{2}\biggr\rfloor\biggr\} + |\ell(h^*)|\geq \lfloor h^*-j^*\rfloor \geq \frac{k^*-j^*+t^*-z^*}{4}.
\] 
Finally, observe by the definition of $h^*$ and the condition $t^*>z^*$ that $h^*-j^* \geq \lfloor (k^*-j^* + 1)/2\rfloor \geq (k^*-j^*)/2$. Hence,
\[
\frac{h^*-j^*}{2} + (2j^*-k^*) \geq \frac{k^*-j^* + (2j^*-k^*)}{4} \geq \frac{j^*}{4} \geq \frac{p}{16}.
\]
Consequently, we have $|\mathcal{J}|\geq |A| \geq \min(p/16, n\tau/2)$. 

\textbf{Case 4}: assume $ k^*-t^* \leq j^*-z^*$ and $k^*+t^* \geq j^*+z^*$. In this case, we have $t^* \geq z^*$.  Let  $A=\{j: j^*-\min\{2j^*-k^*, n\tau/2\}\leq j\leq j^*+\min\{\frac{k^*-j^*}{2}, n\tau/2\}\}$. Noticing that $\ell(j)$ is negative and increasing for $j\in[k^*-j^*, k^*]$, we have 
\[
|\ell(j)|\geq  \biggl|\ell\biggl(\biggl\lfloor\frac{j+k}{2}\biggr\rfloor\biggr)\biggr|\geq t^*-z^*\geq \frac{t^*-z^*+k^*-j^*}{2}.
\]
We have $(k^*-j^*)/2 + (2j^*-k^*) \geq j^*/2 \geq p/8$. Hence $|\mathcal{J}|\geq |A|\geq \min\{p/8, n\tau/2\}$.
\end{proof}

We now state a version of Theorem~\ref{Thm:GeneralConsistency} that does not require the condition $n\tau \geq 2p$.

We also present a general result without the condition $n\tau\geq 2p$:
\begin{theorem}
\label{Thm:Consistency}
Suppose $X\sim P_{j^*,z^*,\mu_0,\mu_1}$ with $\mu_0-\mu_1\in\{-a,a\}^p$. Assuming 
\begin{equation}
\label{Eq:CondSignal}
a^2 \geq c\biggl\{\frac{\sqrt{p}+\log(2pn)}{n\tau\min(p,n\tau)} + \frac{pn\log(2pn)}{n^2\tau^2\min(p,n\tau)^2}\biggr\}.
\end{equation}
Suppose we are using the quadratic statistics defined in equation~\eqref{Eq:quadratic},
we have with probability at least $1-1/(2pn)$ that 
\[
|\hat z - z| + d_G(\hat j , j) \leq C\biggl\{ \frac{\sqrt{p}+\log(2pn)}{a^2\min(p,n\tau)} + \frac{\sqrt{pn \log(2pn)}}{a\min(p,n\tau)}\biggr\}.
\]
\end{theorem}

\begin{proof}
Following the proof of Theorem~\ref{Thm:GeneralConsistency} and  Proposition~\ref{Prop:SetJ}, there exists $\mathcal{J}\subset[p]$ such that $|\mathcal{J}|\geq \min(p,n\tau)/32$, and for each $j\in\mathcal{J}$, we have 
\[
|j-j^*|\leq \frac{n\tau}{2}\quad\text{and}\quad |z^*+d_G(j,j^*)-(\hat z +d_G(j,\hat j))|\geq \min\biggl(\frac{n\tau}{16}, \frac{|z^*-\hat z|+d_G(j^*,\hat j)}{4}\biggr).
\]

Combining equation~\eqref{Eq:Apeak} and Proposition~\ref{Prop:SetJ}, we have that
\begin{align*}
A_{j,z^*+d_G(j,j^*)}-A_{j,\hat z +d_G(j,\hat j)} &\geq \frac{2\theta_j}{3\sqrt{6 n\tau}}\min(|z^*+d_G(j,j^*)-(\hat z +d_G(j,\hat j))|, n\tau/2)\\
&\geq \frac{2\theta_j}{3\sqrt{6 n\tau}}\min\biggl(\frac{|z^*-\hat z|+d_G(j^*,\hat j)}{4}, \frac{n\tau}{16}\biggr).
\end{align*}
Then
\begin{align}
B_{j^*,z^*}-B_{\hat j,\hat z} & \geq \sum_{j\in\mathcal{J}} A_{j,z^*+d_G(j,j^*)}(A_{j,z^*+d_G(j,j^*)}-A_{j,\hat z+d_G(j,\hat j)})\nonumber\\
 &\geq \frac{a^2\min(p, n\tau)}{96\sqrt{6}}\min\biggl(\frac{|z^*-\hat z|+d_G(j^*,\hat j)}{4}, \frac{n\tau}{16}\biggr).\label{Eq:BjLower1}
\end{align}
Combining~\eqref{Eq:BjUpper} in the proof of Theorem~\ref{Thm:GeneralConsistency} and~\eqref{Eq:BjLower1}, and choosing $\delta = 1/(2pn)$, we have with probability at least $1-1/(2pn)$ that 
\begin{equation}
\label{Eq:BasicInequality}
\frac{a^2\min(p, n\tau)}{96\sqrt{6}}\min\biggl(\frac{|z^*-\hat z|+d_G(j^*,\hat j)}{4}, \frac{n\tau}{16}\biggr) \leq 4\sqrt{p+2npa^2\log(2pn)}+8\log(2pn).
\end{equation}
When $c$ is suffiicently large, we have from~\eqref{Eq:CondSignal} that the minimum on the left-hand side of~\eqref{Eq:BasicInequality} is necessarily achieved by the first term. Consequently, we derive from~\eqref{Eq:BasicInequality} that with probability at least $1-1/(2pn)$, we have 
\[
|z^*-\hat z| + d_G(j^*,\hat j) \leq C\biggl\{ \frac{\sqrt{p}+\log(2pn)}{a^2\min(p,n\tau)} + \frac{\sqrt{pn \log(2pn)}}{a\min(p,n\tau)}\biggr\},
\]
as desired.
\end{proof}

With the addition assumption that $n\tau \geq 2p$, for $\tau = \min\{z^*/n, 1-z^*/n\}$, we may establish the following improved version of Proposition~\ref{Prop:SetJ} that can be used to prove Theorem~\ref{Thm:GeneralConsistency}.
\begin{prop}
\label{Prop:SetJ1}
Let $G = C_p$ be a $p$-cycle graph and $\tau = \min\{z^*/n, 1-z^*/n\}$. Assuming that $n\tau\geq 2p$, the following set
\[
   \mathcal{J} := \biggl\{j: \mathrm{sgn}(z^*-t^*)\ell(j)\geq \frac{|z^*-t^*|+d_G(j^*,k^*)}{4}\biggr\}           
\]
has cardinality at least $p/8$.
\end{prop}

\begin{proof}
Following the proof of Proposition~\ref{Prop:SetJ}, we may assume without loss of genrality that $j^* = \lceil p/2 \rceil$ and $k^*\geq j^*$, which imlpies that $\ell(j)$ takes the form given in~\eqref{Eq:ell}. We then prove the result by considering four cases as in the proof of Proposition~\ref{Prop:SetJ}.

\textbf{Case 1}: assume $k^*-t^* \geq j^*-z^*$ and $k^*+t^* \leq j^*+z^*$. In this case, we have $t^*\leq z^*$ and $k^*-j^*\leq z^*-t^*$. Hence $\ell(j) \geq 0$ for all $j$.  Notice that $\ell(j)$ is an non-decreasing function of $j$ for $j\geq j^*$. Then, for all $j$ such that $j^* +  (k^*-j^*)/4 \leq j\leq p$, we have 
\begin{equation*}
\ell(j) \geq \ell\bigl(j^*+\lceil (k^*-j^*)/4\rceil\bigr) \geq  z^*-t^*+\frac{1}{2}(k^* -j^*) \geq \frac{z^*-t^*+k^*-j^*}{4}.    
\end{equation*}
Consequently, in this case, $|\mathcal{J}| \geq p - j^* - \lceil (k^*-j^*)/4\rceil + 1 \geq p/4$ as required.

\textbf{Case 2}: assume $k^*-t^* \geq j^*-z^*$, $k^*+t^* \geq j^*+z^*$ and $z^*\geq t^*$. In this case, $k^*-j^*\geq z^*-t^*\geq 0$.
We define $h^*:=j^*+\lceil \frac{(k^*-j^*)-(z^*-t^*)}{2}\rceil$, and observe that $j^*\leq h^*\leq k^*$, $k^*-h^*\geq (k^*-j^*)/2$ and $\ell(h^*) \in\{0,1\}$, and that $\ell(j)$ is increasing for $j \in [h^*, p]$.  Then, for all $j$ such that $(k^*+h^*)/2 \leq j\leq p$, we have 
\begin{align*}
\ell(j) &\geq \ell\biggl(\biggl\lceil \frac{h^*+k^*}{2}\biggl\rceil\biggr) = \ell(h^*) + 2\biggl\lceil \frac{k^*-h^*}{2}\biggr\rceil \geq \ell(h^*) + (k^* - h^*) \\
&\geq \ell(h^*) +   \frac{z^*-t^*+k^*-j^* - \ell(h^*)}{2} \geq \frac{z^*-t^*+k^*-j^*}{2},
\end{align*}
where in the penultimate inequality, we have used the property that $\lceil\frac{(k^*-j^*) - (z^*-t^*)}{2}\rceil = \frac{(k^*-j^*) - (z^*-t^*) + \ell(h^*)}{2}$. Consequently, in this case, $|\mathcal{J}| \geq p-\lceil(k^*+h^*)/2\rceil+1$. The right-hand side is a decreasing function of $k^*$. Hence, using the fact that $k^*\leq p$ and $h^*\leq (j^*+k^*)/2$, we have $|\mathcal{J}|\geq p/8$ as desired.

\textbf{Case 3}: assume $k^*-t^* \geq j^*-z^*$, $k^*+t^* \geq j^*+z^*$ and $t^* > z^*$. We define $h^*:=j^*+\lfloor \frac{(k^*-j^*)+(t^*-z^*)}{2}\rfloor$. Observe that $\ell(h^*) \in\{0,-1\}$ and $h^* = \frac{(k^*+j^*)+(t^*-z^*)}{2} - \frac{\ell(h^*)}{2}$. In this case, $k^*-j^*\geq t^*-z^*\geq 0$. For all $j$ such that $k^*-j^* \leq j\leq \frac{h^*+j^*}{2}$, $\ell(j)$ is negative and increasing, satisfying
\begin{align*}
-\ell(j)&\geq -\ell\biggl(\biggl\lfloor\frac{j^*+h^*}{2}\biggr\rfloor\biggr) = 2 \biggl\{h^*-\biggl\lfloor\frac{j^*+h^*}{2}\biggr\rfloor\biggr\} -\ell(h^*)\\
&\geq  h^*-j^* - \ell(h^*) \geq \frac{k^*-j^*+t^*-z^*}{2}.
\end{align*}
Finally, observe by the definition of $h^*$ and the condition $t^*>z^*$ that $h^*-j^* \geq \lfloor (k^*-j^* + 1)/2\rfloor \geq (k^*-j^*)/2$. Hence, we have 
\[
|\mathcal{J}| \geq \frac{h^*+j^*}{2} - (k^*-j^*) \geq \frac{h^*-j^*}{2} + (2j^*-k^*) \geq \frac{k^*-j^* + (2j^*-k^*)}{4} \geq \frac{j^*}{4} \geq \frac{p}{8}.
\]

\textbf{Case 4}: assume $ k^*-t^* \leq j^*-z^*$ and $k^*+t^* \geq j^*+z^*$. In this case, we have $t^* \geq z^*$. For all $j$ such that $k^*-j^*\leq j\leq \frac{k^*+j^*}{2}$, we note that  $\ell(j)$ is negative and increasing, satisfying
\[
-\ell(j)\geq  -\ell\biggl(\biggl\lfloor\frac{j^*+k^*}{2}\biggr\rfloor\biggr)\geq t^*-z^*\geq \frac{t^*-z^*+k^*-j^*}{2}.
\]
Hence, We have $|\mathcal{J}|\geq (k^*+j^*)/2 - (k^* - j^*) \geq j^*/2 \geq p/4$, completing the proof.
\end{proof}

For the case when we are using the linear statistics, we provide the following result of the difference between the sum of $E_{j,z^*+d_G(j,j^*)}$ and $E_{j,t^*+d_G(j,k^*)}$ for coordinates in set $\mathcal{J}_{t^*,k^*}(C_1)$.
\begin{prop}
\label{Prop:Variance}
Fix $z^*\in[n-1]$ and $j^*\in[p]$. If $n\tau \geq 2p$, then there exists a universal constant $C > 0$ and an event with probability at least $1-1/(pn)$, such that on this event, for all $t^*\in[n-1]$, $k^*\in[p]$ and $\mathcal{J}_{t^*,k^*}(C_1)\subseteq [p]$ defined in~\eqref{Eq:mathcalJ}, we have

\[
\sum_{j\in\mathcal{J}_{t^*,k^*}(C_1)}(E_{j,z^*+d_G(j,j^*)}-E_{j,t^*+d_G(j,k^*)})\leq C \sqrt{\frac{|\mathcal{J}_{t^*,k^*}(C_1)|(|z^*-t^*|+d_G(k^*,j^*))\log(pn)}{n\tau}}.
\]
\end{prop}
\begin{proof}
First, we claim that if $|z^*-t^*| + d_G(k^*,j^*) \geq n\tau/2$, then the conclusion holds trivially. To see this, we note that $\sum_{j\in[p]} E_{j, z^*+d_G(j,j^*)} \sim N(0, p)$ and $\sum_{j\in[p]} E_{j, t^*+d_G(j,k^*)} \sim N(0, p)$. Taking a union bound over $t^*$ and $k^*$, there is an event with probability at least $1-1/(pn)$ such that 
\[
 \max\biggl\{\sum_{j\in \mathcal{J}_{t^*,k^*}(C_1)} E_{j, z^*+d_G(j,j^*)}, \max_{t^*\in[n-1], k^*\in[p]} \sum_{j\in\mathcal{J}_{t^*,k^*}(C_1)} E_{j, t^*+d_G(j,k^*)}\biggr\} \leq 2\sqrt{p\log(pn)}.
\]
So it suffices to take $C = 2\sqrt{2}$ for the desired conclusion to hold. Hence, we may assume without loss of generality that $|z^*-t^*| + d_G(k^*,j^*) < n\tau/2$.

We control $\sum_{j\in \mathcal{J}_{t^*,k^*}(C_1)} (E_{j,z^*+d_G(j,j^*)} - E_{j,t^*+d_G(j,k^*)})$ for fixed $t^*\in[n-1]$, $k^*\in[p]$. For simplicity of notation, we denote $z_j := z^*+d_G(j,j^*)$ and $t_j: = t^*+d_G(j,k^*)$. Note that $\sum_{j\in \mathcal{J}_{t^*,k^*}(C_1)} (E_{j,z^*+d_G(j,j^*)} - E_{j,t^*+d_G(j,k^*)})$ is a sum of $|\mathcal{J}_{t^*,k^*}(C_1)|$ independent normal random variables. Hence, we start by controlling the variance of each summand. We consider first the case where $t_j \leq z_j$. From the definition of the CUSUM transformation, we can write
\begin{align}
     E_{j,z_j}-E_{j,t_j}
    &= \sqrt{\frac{n}{z_j(n-z_j)}}\bigg(\frac{z_j}{n}\sum_{r=1}^n  W_{j,r}-\sum_{r=1}^{z_j} W_{j,r}\bigg)\nonumber\\
    &\hspace{4cm} -\sqrt{\frac{n}{t_j(n-t_j)}}\bigg(\frac{t_j}{n}\sum_{r=1}^n W_{j,r}-\sum_{r=1}^{t_j} W_{j,r}\bigg)\nonumber\\
    &=  \sqrt{\frac{n}{z_j(n-z_j)}}\bigg(\frac{z_j-t_j}{n}\sum_{r=1}^n W_{j,r}-\sum_{r=t_j+1}^{z_j} W_{j,r}\bigg)\nonumber\\
    &\hspace{2cm}+\bigg(\sqrt{\frac{n}{z_j(n-z_j)}}-\sqrt{\frac{n}{t_j(n-t_j)}}\bigg) \bigg(\frac{t_j}{n}\sum_{r=1}^n W_{j,r}-\sum_{r=1}^{t_j}  W_{j,r}\bigg).
\end{align}
By the mean value theorem, there exists $\xi\in[t_j,z_j]$, such that
\[
\bigg(\sqrt{\frac{n}{z_j(n-z_j)}}-\sqrt{\frac{n}{t_j(n-t_j)}}\bigg) \leq (z_j-t_j)\bigg|\frac{\xi}{n}-\frac{1}{2}\bigg|\bigg(\frac{n}{\xi(n-\xi)}\bigg)^{3/2}\leq \frac{\sqrt{2}(z_j-t_j)}{\min(\xi,n-\xi)^{3/2}}
\]
Also, we observe that:
\[
\frac{t_j}{n}\sum_{r=1}^n W_{j,r}-\sum_{r=1}^{t_j}  W_{j,r}=\sum_{r=t+1}^n W_{j,r}-\frac{n-t_j}{n} \sum_{r=1}^n W_{j,r}.
\]
Since $\sum_{r=1}^n W_{j,r}$ and $\sum_{r=t_j+1}^{z_j} W_{j,r}$ are positively corrected with each other, we have 
\begin{align*}
    \mathbbm{V}( E_{j,z_j}-E_{j,t_j})&\leq \frac{2n}{z_j(n-z_j)} \bigg(\frac{(z_j-t_j)^2}{n}+z_j-t_j\bigg)\\
    &\hspace{2cm}+\frac{4(z_j-t_j)^2}{\min(\xi,n-\xi)^3}\min\bigg(\frac{t_j^2}{n}+t_j,n-t_j+\frac{(n-t_j)^2}{n}\bigg)\\
    &\leq 4(z_j-t_j)\bigg(\frac{1}{z_j}+\frac{1}{n-z_j}\bigg)+\frac{8(z_j-t_j)^2}{\min(t_j,n-z_j)^2}\max\bigg(1,\frac{n-t_j}{n-z_j}\bigg)
\end{align*}
Since $n\tau \geq 2p$, we have $|z_j - z^*| = d_G(j - j^*) \leq p \leq n\tau/2$ and consequently $n\tau \leq z_j\leq n-n\tau/2$. Also, by~\eqref{Eq:ell}, we have $z_j-t_j<|z^*-t^*| + d_G(k^*,j^*)<n\tau/2$, so $n\tau/2 \leq t_j \leq n-n\tau$. Thus, for some universal constant $C>0$, we have 
\begin{align}
\label{Eq:Var_j}
\mathbbm{V}( E_{j,z_j}-E_{j,t_j})&\leq \frac{8(z_j-t_j)}{n\tau} + \frac{4n\tau(z_j-t_j)}{(n\tau/2)^2}\biggl(1+\frac{n\tau/2}{n\tau/2}\biggr) \nonumber\\
&\leq \frac{C(z_j-t_j)}{n\tau}\leq \frac{C(|z^*-t^*|+d_G(k^*,j^*))}{n\tau}.
\end{align}
If $z_j > t_j$, a symmetric argument will show that the same variance bound as in~\eqref{Eq:Var_j} holds. Therefore, 
\[
\mathbbm{V}\bigg(\sum_{j\in\mathcal{J}_{t^*,k^*}(C_1)} (E_{j,z_j}-E_{j,t_j})\bigg)
\leq \frac{C|\mathcal{J}_{t^*,k^*}(C_1)|(|z^*-t^*|+d_G(k^*,j^*))}{n\tau}
\]
Then, for a fixed $t^*$ and $k^*$, we have that 
\begin{align*}
 &\mathbb{P}\bigg(\sum_{j\in\mathcal{J}_{t^*,k^*}(C_1)} (E_{j,z^*+d_G(j,j^*)}-E_{j,t^*+d_G(j,k^*)})\geq \\
 &\hspace{4cm} \sqrt{\frac{4C|\mathcal{J}_{t^*,k^*}(C_1)|(|z^*-t^*|+d_G(k^*,j^*))\log (pn)}{n\tau}}\bigg)\leq \frac{1}{(pn)^2}.
\end{align*}
The desired conclusion then follows by taking a union bound over $t^*\in[n-1]$ and $k^*\in[p]$. 
\end{proof}

\bibliographystyle{custom}
\bibliography{biblio}

\end{document}